\documentclass[11pt]{article}
\usepackage[utf8]{inputenc}
\usepackage{style}
\usepackage{framed}
\usepackage{nicefrac}
\usepackage[parfill]{parskip}
\usepackage{setspace}
\usepackage{hyperref}
\usepackage{mdwlist}
\usepackage{mdframed}
\usepackage[ruled,linesnumbered]{algorithm2e}
\usepackage[utf8]{inputenc}
\usepackage{style}
\usepackage{mdframed}
\usepackage{mdwlist}
\usepackage{nicefrac}
\usepackage[parfill]{parskip}
\usepackage{setspace}
\usepackage[colorlinks]{hyperref}
\usepackage[ruled,linesnumbered]{algorithm2e}
\usepackage{amsmath}
\usepackage{amssymb}
\usepackage{complexity}

\newcommand{\ug}{\text{\sc Unique-Games}}

\newcommand{\Ex}{\E}
\newcommand{\cH}{\mathcal{H}}
\renewcommand{\cP}{{\cal P}}

\allowdisplaybreaks

\newcommand{\defeq}{\overset{\rm def}{=}}
\usepackage{microtype}

\usepackage{thm-restate}

\usepackage{boxedminipage}

\renewcommand{\ug}{{\sc Unique~Games}}

\newcommand{\Inf}[2]{{\sf Inf}_{#1}\left[#2\right]}
\newcommand{\non}{\nonumber}

\newcommand{\cE}{{\mathcal{E}}}

\newcommand{\wh}[1]{\widehat{#1}}

\newcommand{\THRESH}{{\mathcal{THRESH}}}
\newcommand{\bv}{\mathbf{v}}
\newcommand{\br}{\mathbf{r}}

\providecommand{\eps}[0]{\ensuremath{\varepsilon}\xspace}

\newcommand{\one}{\mathbbm{1}}

\newcommand{\ccut}{~\bar{\sf Cut}_\rho}

\newcommand{\bg}{\bar{g}}
\newcommand{\dsp}{\displaystyle}
\renewcommand{\R}{\mathbbm{R}}
\newcommand{\bmu}{\bar{\mu}}
\renewcommand{\cS}{\mathcal{S}}
\newcommand{\bh}{\bar{h}}
\newcommand{\lGamma}{\underline{\Gamma}}
\allowdisplaybreaks

\title{On the Approximability of Max-Cut on $3$-Colorable Graphs and Graphs with Large Independent Sets\thanks{Companion code: \url{https://github.com/nenghuang/max-cut-independent-set}}}
\author{Suprovat Ghoshal\thanks{Indiana University.} 
\and Neng Huang\thanks{TU Dresden. Work was mostly done when NH was affiliated with the University of Michigan. NH has been funded by the European Union (Project POCOCOP, ERC Synergy Grant 101071674). Views and opinions expressed are however those of the author(s) only and do not necessarily reflect those of the European Union or the European Research Council Executive Agency. Neither the European Union nor the granting authority can be held responsible for them.} 
\and Euiwoong Lee\thanks{University of Michigan. Supported in part by NSF award CCF-2236669.} 
\and Konstantin Makarychev\thanks{Northwestern University. Supported by NSF awards CCF-1955351 and EECS-2216970.} 
\and Yury Makarychev\thanks{Toyota Technological Institute at Chicago. Supported by NSF awards CCF-1955173 and ECCS-2216899.}}
\begin{document}

\begin{titlepage}
\maketitle
\begin{abstract}
	Max-Cut is a classical graph-partitioning problem where given a graph $G = (V,E)$, the objective is to find a cut $(S,S^c)$ which maximizes the number of edges crossing the cut. In a seminal work, Goemans and Williamson gave an $\alpha_{GW} \approx 0.87856$-factor approximation algorithm for the problem, which was later shown to be tight by the work of Khot, Kindler, Mossel, and O'Donnell. Since then, there has been a steady progress in understanding the approximability at even finer levels, and a fundamental goal in this context is to understand how the structure of the underlying graph affects the approximability of the Max-Cut problem. \\

In this work, we investigate this question by exploring how the chromatic structure of a graph affects the Max-Cut problem. In particular, it is known that $2$-colorable graphs are perfect instances of Max-Cut and they are polynomial time solvable, and various works have developed robust variants of this observation in the 99\%-regime, i.e., when the graph is nearly $2$-colorable. In this work, we study this question in the 1\%-regime instead, where the graph has significantly weaker structure. Specifically, we consider the settings where the graph is {\em $3$-colorable} or contains a {\em large independent set} (in terms of its volume), and we ask if Max-Cut is easier to approximate in such instances. \\

Our main contributions in this context are as follows: \\[-10pt]
\begin{itemize}
	\item We show Max-Cut is $\alpha_{GW}$-hard to approximate for $3$-colorable graphs. \\[-10pt]
	\item We identify a natural threshold $\alpha^*$ such that the following holds. 
	\begin{itemize}
		\item For graphs which contain an independent set of size up to $\alpha^*$, Max-Cut continues to be $\alpha_{GW}$-factor hard to approximate.
		\item For any graph that contains an independent set of size $> \alpha^*$, there exists an efficient $>\alpha_{GW}$-approximation algorithm for Max-Cut. \\
	\end{itemize} 
\end{itemize}

Our hardness results are derived using various analytical tools and novel variants of the Majority-Is-Stablest theorem, which might be of independent interest. Our algorithmic results are based on a novel SDP relaxation, which is then rounded and analyzed using interval arithmetic. Interestingly, our setting with the promise of an unknown independent set serves as a natural example of valued PCSPs recently proposed by Barto, Butti, Kazda, Viola, and Živný.

\end{abstract}
\end{titlepage}

\tableofcontents
\newpage
\section{Introduction}

Max-Cut is a classical combinatorial optimization problem of longstanding interest in both theory and practice. Formally, in Max-Cut, one is given a graph $G = (V,E)$, and the objective is to find a cut $(S,S^c)$ which maximizes $|E[S,S^c]|$, i.e., the number of edges crossing the cut. It arises as a natural computational task in various areas such as operations research~\cite{FIA11social}, statistical physics~\cite{JKM18}, machine learning~\cite{JW08graphical}, among many others. 

As a result, the complexity of the Max-Cut problem has been studied extensively from various perspectives, as evidenced by the enormous literature on it. It was one of Karp's original $21$ \NP-hard problems~\cite{Karp75}, which motivated the study of polynomial-time approximation algorithms for Max-Cut. In a landmark breakthrough, Goemans and Williamson~\cite{GW94} gave an approximation algorithm  with ratio $\alpha_{GW} \approx 0.87856$  based on semi-definite programming (SDP). This approximation guarantee was then shown to be tight by Khot, Kindler, Mossel, and O'Donnell~\cite{KKMO07}, who showed an $(\alpha_{GW} + \epsilon)$-hardness for approximating Max-Cut~\cite{KKMO07} for any constant $\epsilon > 0$, assuming the Unique Games Conjecture (UGC)~\cite{Khot02a}. In an even further refinement of these results, O'Donnell and Wu~\cite{OW08} and Raghavendra~\cite{Rag08} independently gave UGC-based optimal algorithms and hardness results for the entire approximation curve of Max-Cut.
 
Despite these tight approximation bounds for Max-Cut, our quest to understand the landscape of approximation for Max-Cut (and other optimization problems in general) still remains far from complete as we strive to obtain a more fine-grained view of its complexity. In this regard, a broad ongoing effort revolves around understanding how the structure of the input graph affects the approximability of Max-Cut. For instance, Max-Cut has been shown to admit improved approximation ratios for various classes of structured instances, such as bounded-degree graphs~\cite{FKL02, hsieh2023approximating}, graphs with large spectral gap~\cite{AKKSTV08,Kol10,MM10} or low threshold rank~\cite{BRS11,GS11}, graphs with a linear number of edge-disjoint triangles~\cite{GLR25, ADKL25}, dense graphs~\cite{AKK95}, interval graphs~\cite{ADKL25}, and graphs with a low-rank SDP solution~\cite{AZ05}. These lines of work serve to bridge the gap between existing worst-case approximation guarantees for the problem and improved guarantees one can achieve when the instances have more structure.

One such natural source of graph structure for Max-Cut comes from the chromatic number. It is well known that perfectly complete Max-Cut instances are bipartite (i.e., 2-colorable), and Max-Cut is polynomial-time solvable on such graphs. Furthermore, if a graph is nearly bipartite (up to edge or vertex deletions), then one can efficiently recover a near-optimal cut~\cite{CMM06,Tre09,ACMM05,GL21}. 
This makes it natural to ask whether 3-colorability, the next simplest chromatic promise, can help us do better than Goemans-Williamson algorithm.
\begin{itemize}
	\item[$\triangleright$] {\it Can we get better than $\alpha_{GW}$-approximation for Max-Cut on $3$-colorable graphs?} 
\end{itemize}

As a further relaxation of the condition of having small chromatic number, we also explore the approximability of Max-Cut in graphs which contain large independent sets. For instance, bipartite graphs, as well as graphs with bounded degree, contain independent sets of linear size. As mentioned earlier, both classes admit better than $\alpha_{GW}$-approximation for Max-Cut. More generally, having large independent sets is a natural structural condition closely related to the condition of having bounded chromatic number\footnote{For instance, approximation algorithms for coloring problems~\cite{KMS98,KT17,BG21,bansal2026improved} often employ sub-routines for finding large independent sets to color large subsets of vertices with small number of colors. In the other direction, graphs with bounded chromatic number are guaranteed to contain linear-sized independent sets.}. This motivates us to ask whether graphs with large independent sets are easier instances for Max-Cut. There is, however, one caveat: one can always pad a graph with isolated vertices to increase the size of the maximum independent set. For this reason, we instead measure the size of an independent set by the fraction of edges incident on it, or equivalently by its \emph{volume}, defined as the sum of (weighted) degrees of its vertices. This leads us to ask:
\begin{itemize}
	\item[$\triangleright$] {\it Can we get better than $\alpha_{GW}$-approximation for Max-Cut on graphs that contain an independent set of large volume?} 
\end{itemize}

\subsection{Our Results}

In this paper, we address the questions stated above and obtain both positive and negative results. Before we proceed, let us introduce some notation that will be used for stating them. Throughout the paper we will use $\rho^*$ to denote the value
\[
\rho^* := \argmin_{\rho \in [-1,1]}\frac{2\arccos \rho}{\pi(1 - \rho)} \approx -0.68915,
\]
i.e., it is the critical inner-product for which the Goemans-Williamson hyperplane rounding achieves the worst-case performance, and let
\[
\alpha_{GW} := \frac{2\arccos \rho^*}{\pi(1 - \rho^*)} \approx 0.87856
\]
denote the corresponding worst-case approximation ratio. In the hardness results stated below, our YES and NO cases always distinguish between when the optimal max-cut is of size at least $c^*:= \frac{1 - \rho^*}{2}\approx 0.8445$ and the case when the optimal max-cut is of size at most $s^* := \frac{\arccos \rho^*}{\pi} \approx 0.742$, where $c^*$ and $s^*$ are the SDP value and the (expected) rounded value on an edge whose vectors have inner-product $\rho^*$. 

We begin with our first result, which says that when the input graph is $3$-colorable, the Max-Cut problem is hard to approximate beyond the standard $\alpha_{GW}$ factor.

\begin{theorem}					\label{thm:3-col}
	Assuming UGC\footnote{See Section \ref{sec:ugc} for a formal description of the Unique Games Conjecture (UGC).}, the following holds for any small constant $\eta > 0$. Given a weighted $3$-colorable graph $G = (V,E,w)$ it is NP-hard to distinguish between the following cases:
	\begin{itemize}
		\item {\bf YES Case}. There exists a set $S \subseteq V$ such that $w(S,S^c) \geq \frac{1 - \rho^*}{2}-\eta$. 
		\item {\bf NO Case}. The Max-Cut of $G$ is at most $\frac{\arccos \rho^*}{\pi} + \eta$.
	\end{itemize}
	Furthermore, the hardness result holds even when the underlying $3$-coloring is known to the algorithm.
\end{theorem}

In the above theorem, for a subset $S \subseteq V$, the notation $w(S,S^c)$ denotes the weight of the edges crossing the cut $(S,S^c)$ in $G$.
The above result stands in sharp contrast to the $2$-colorable setting where Max-Cut is polynomial-time solvable. Interestingly, the above theorem states the knowledge of the $3$-partition cannot be used to get better than $\alpha_{GW}$ approximation. On the other hand, we will now state a collection of results which show that the approximability of Max-Cut can be much more nuanced under the promise of having a large independent set i.e., it will depend on whether the independent set is known or hidden from the algorithm, as well as the relative volume of the independent set itself. Among these, the first result shows that a similar hardness holds when the graph contains a (known) independent set with large relative volume.

\begin{theorem}				\label{thm:iset-2}
	Assuming UGC, the following holds for any small constant $\eta > 0$. Let $G = (V,E,w)$ be a weighted graph that contains an independent set $I \subseteq V$ satisfying
	\[
	w(I,I^c) = \frac{\arccos \rho^*}{\pi} - \eta.
	\] 
	Then given such a graph, it is NP-hard to distinguish between the following cases:
	\begin{itemize}
		\item {\bf YES Case}. There exists a set $S \subseteq V$ such that $w(S,S^c) \geq \frac{1 - \rho^*}{2}-\eta$.
		\item {\bf NO Case}. The Max-Cut of $G$ is at most $\frac{\arccos \rho^*}{\pi} + \eta$.
	\end{itemize}
	Moreover, the hardness holds even when the independent set $I$ is also provided as input to the algorithm.
\end{theorem} 

We make some brief observations regarding the above result. As mentioned previously, in the setting of the above theorem (as well as the theorems stated below), we are measuring the size of the independent set via the fraction of edges that are incident on the independent set; up to renormalization, this is equivalent to the volume of the independent set, which is a more natural way of measuring the size of a set when the graph is irregular. Now, the above theorem states that the $\alpha_{GW}$-factor hardness also holds for instances which contain independent sets of relative volume at most $s^*/2 \approx 0.371$. Furthermore, the above result is {\em tight} in the following sense: if the independent set had more than $s^*$-fraction of edges incident on it, then since the independent set $I$ is known to the algorithm, it can trivially achieve better than $\alpha_{GW}$-approximation by simply returning the cut $(I,I^c)$. 

On the other hand, when the independent set is not revealed to the algorithm, we can prove $\alpha_{GW}$-hardness for instances which contain even larger independent sets. We formally state this result in the following theorem.

\begin{theorem}		\label{thm:i-set}
Assuming UGC, the following holds for any small constant $\eta > 0$. Let $\alpha^* = \frac{2\rho^*}{\rho^* - 1} \approx 0.81597$. Then, given a weighted graph $G = (V,E,w)$ it is NP-hard to distinguish between:
\begin{itemize}
	\item {\bf YES Case}. There exists $I \subseteq V$ such that $I$ is an independent set and $w(I,I^c) \geq \alpha^* - \eta$. Moreover, there exists a set $S \subseteq V$ such that $w(S,S^c) \geq \frac{1 - \rho^*}{2} - \eta$.
	\item {\bf NO Case}. The Max-Cut of $G$ is at most $\frac{\arccos \rho^*}{\pi} + \eta$.
\end{itemize}
\end{theorem}

In the above theorem, the quantity $\alpha^* \approx 0.81597$ actually corresponds to a threshold which captures the tradeoff between the eigenvalues and the size of the independent set for a natural family of gadgets; we point the readers to Section \ref{sec:iset-overview} for a discussion which explains the origins of this threshold in more detail. We also give a ``sharp converse'' to the above theorem, in which we show that if a graph indeed contains a (hidden) independent set which has more than $\alpha^*$-fraction of edges incident on it, then we can achieve better than $\alpha_{GW}$-approximation for Max-Cut.

\begin{theorem}\label{thm:algo}
    Let $G = (V, E)$ be a graph. Assume that there exists $I \subseteq V$ such that $I$ is an independent set and $w(I,I^c) > \alpha^*$ where $\alpha^* \approx 0.81597$, then we can find a set $S \subset V$ in polynomial time such that $w(S,S^c) > \frac{\arccos \rho^*}{\pi}$.
\end{theorem}

\subsection{Discussion}

Our first two results, Theorem \ref{thm:3-col} and Theorem \ref{thm:iset-2}, show that under the promise that the graph is $3$-colorable, or it contains a somewhat large independent set, the Max-Cut problem is still hard to approximate beyond the $\alpha_{GW}$-factor, even when the $3$-coloring or the independent set is revealed to the algorithm. In the case where the underlying independent set is hidden, we show that the $\alpha_{GW}$-hardness continues to hold for graphs with even larger (hidden) independent sets (Theorem \ref{thm:i-set}). Finally, as a sharp converse to Theorem \ref{thm:i-set}, in Theorem \ref{thm:algo} we show that if the underlying graph contains an independent set that is slightly larger than the given threshold $\alpha^* \approx 0.81597$, one can get better than $\alpha_{GW}$-approximation. 

We point out that Theorem \ref{thm:algo} also serves as a weak converse to Theorem \ref{thm:3-col} in the following sense: in Theorem \ref{thm:3-col}, the underlying $3$-coloring of the hard instances produced by the reduction is perfectly balanced, that is, the three color classes all have equal size. On the other hand, if the input graph admits a highly unbalanced 3-coloring where one of the color classes is incident on more than an $\alpha^*$-fraction of edges, then Theorem \ref{thm:algo} yields an  approximation ratio better than $\alpha_{GW}$. 

Overall, our results explore the approximability of Max-Cut in regimes where the input graph has much weaker structure than being bipartite. We conclude by suggesting an interesting connection between the setting of Theorems~\ref{thm:i-set} and~\ref{thm:algo} and the recently introduced \emph{Valued PCSP} framework~\cite{barto2024algebraic}. Informally, in a valued PCSP, we are given two ways to evaluate a MAX CSP instance, such that the existence of a good solution in the first evaluation method implies the existence of a good solution in the second. Our setting seems a natural example of valued PCSPs since we may think of each edge as one (valued) Boolean constraint $P$, where in the independent set case we have $P(0, 0) = 0, P(1, 0) = P(0, 1) = 1, P(1, 1) = -\infty$ (where 1 means that the variable is in the independent set, so the endpoints can't both be 1), and in the MAX CUT case we have $P(0, 0) = (1, 1) = 0, P(1, 0) = P(0, 1) = 1$. It would be interesting to explore whether our analytic methods in this paper can be applied more broadly to valued PCSPs, or in the other direction, whether the algebraic tools developed in~\cite{barto2024algebraic} can shed more light on problems discussed in this paper. 

\subsection{Related Works}

We round off this section with a brief discussion on some additional related works. The related problems of coloring a graph with small number of colors, and that of finding large independent sets, have been studied quite extensively in the literature~\cite{KMS98,DRM06,DS10,KT17,GS20,BG21}. In particular, the current state-of-the-art for coloring $3$-colorable graphs is by Kawarabayashi, Thorup, and Yoneda~\cite{KTY24}, who gave an efficient algorithm that can color such graphs using $\tilde{O}(n^{0.19747})$-colors\footnote{A recent preprint~\cite{bansal2026improved} announced an improved algorithm which uses $O(n^{0.19539})$ colors.}. There has also been recent progress on finding near-optimal colorings and independent sets in the setting when the underlying graph is a one-sided expander~\cite{BHK25,BRYS25}, or has low-threshold rank~\cite{Hsieh25coloring}. Another interesting related work is that of Nakajima and Živný~\cite{NZ25}, which explores the interplay between graph-cut problems and graph coloring using the lens of Promise CSPs; in particular, they show that given a Max-Cut instance $G$, one can efficiently find a $3$-partition which cuts ${\sf SdpVal}(G)$ fraction of edges in $G$, where ${\sf SdpVal}(G)$ is the optimal value of the standard SDP relaxation for Max-Cut on $G$. 

\section{Technical Overview}

In this section, we give a technical overview of our results. As is standard, our hardness results are derived using PCP-based gadget reductions: they follow the standard template of first designing an appropriate {\em dictatorship test} (i.e., an inner-verifier), which when combined with \ug~as the outer-verifier\footnote{See Section \ref{sec:ugc} for a formal description of the \ug~problem.} immediately yields the corresponding hardness. Formally, dictatorship tests in our reductions are defined in the following way. Let $\Omega$ be a finite set and $\mu$ be a distribution on $\Omega \times \Omega$. Then a $(c,s)$-dictatorship test is a weighted graph $G^R = (\Omega^R, \Omega^R \times \Omega^R, \mu^R)$, where $R$ is a large integer, such that:
\begin{itemize}
	\item The vertex set of $G^R$ is $\Omega^R$, and the edges in $\Omega^R \times \Omega^R$ are weighted according to the $R$-wise product measure $\mu^R$.
	\item {\bf Completeness}. There exists a (dictator) set $S^* \subset \Omega^R$ that cuts at least $c$-fraction of edges in $G^R$.
	\item {\bf Soundness}. Every quasirandom set\footnote{A set $S \subset \Omega^R$ is said to be quasirandom if for every $j \in [R]$ the $j^{th}$ influence of the indicator function $\one_S$ is small.} $S \subset \Omega^R$ cuts at most $s$-fraction of edges in $G^R$.
\end{itemize} 

We refer to $c$ and $s$ as the completeness and soundness parameters of the dictatorship test gadget. By standard techniques~\cite{KKMO07,Rag08}, one can use a $(c,s)$-dictatorship test to immediately show a $c$ vs. $s$ hardness for Max-Cut assuming UGC. Henceforth, we will just focus our discussion on the design and analysis of the dictatorship gadgets for our hardness reductions.

\subsection{Review of KKMO's $\alpha_{GW}$-hardness}

We begin with a quick review of \cite{KKMO07}'s UGC-based $\alpha_{GW}$-hardness for Max-Cut\footnote{Readers who are familiar with \cite{KKMO07}'s techniques for proving the $\alpha_{GW}$-hardness may choose to skip ahead to Section \ref{sec:3col-overview}.}. As mentioned above, the key component of the $\alpha_{GW}$ is an appropriately-chosen dictatorship test gadget, which is the $\rho^*$-correlated noisy hypercube $H^R_{\rho^*}$. The vertex set of $H^R_{\rho^*}$ is $\{0,1\}^R$, i.e., the set of all $R$-length Boolean strings. The edge weight on a pair of vertices $(x,y) \in \{0,1\}^R \times \{0,1\}^R$ is the probability that it is sampled by the procedure described in Figure \ref{fig:corr}. 

\begin{figure}[ht!]
\begin{mdframed}
	\begin{itemize}
		\item Sample the string $x$ uniformly from $\{0,1\}^R$.
		\item Sample $y$ by sampling each bit independently from the following distribution:
		\[
		y(i) = 
		\begin{cases}
			x(i) & \text{ with probability}~\frac{1 + \rho^*}{2} \\
			1 - x(i) & \text{ with probability}~\frac{1 - \rho^*}{2}.
		\end{cases}
		\]
	\end{itemize}
\end{mdframed}	
\caption{Sampling a $\rho^*$-Correlated Pair}
\label{fig:corr}
\end{figure}
In other words, a random draw of an edge from $H^R_{\rho^*}$ corresponds to sampling a pair of $\rho^*$-correlated strings from $\{0,1\}^R$. The completeness of the gadget can be established easily; consider the dictator cut $S = \{x \in \{0,1\}^R : x(1) = 1\}$. It can be easily seen that:
\[
\Pr_{(x,y) \sim H^R_{\rho^*}}\Big[\text{$S$ cuts $(x,y)$}\Big]
 = \Pr_{(x,y) \sim H^R_{\rho^*}} \Big[x(1) \neq y(1)\Big] = \frac{1 - \rho^*}{2} = c^*,
\]
i.e., the set $S$ cuts $\frac{1 - \rho^*}{2}$-fraction of edges in $H^R_{\rho^*}$. For arguing soundness, one needs to show that every quasirandom cut has (relatively) small number of edges crossing it. To that end, as a first step, a folding argument can be used to show that it suffices to consider the case where the subsets are balanced (see Footnote \ref{fnt:balanced}). Then, one can use the {\em Majority-Is-Stablest} theorem (stated below), to bound the maximum cut that can be obtained by balanced quasirandom sets.

\begin{theorem}[Majority is Stablest~\cite{MOO10} restated]			\label{thm:mis}
	For every balanced quasirandom set $S \subseteq \{0,1\}^n$, we have
	\[
	\Pr_{(x,y) \sim H^R_{\rho^*}}\Big[\text{$S$ cuts $(x,y)$}\Big] \leq \frac{\arccos \rho^*}{\pi} + o(1) = s^* + o(1).
	\] 
\end{theorem}

The completeness and soundness guarantees taken together imply that Max-Cut is hard to approximate beyond $\alpha_{GW} = s^*/c^* \approx 0.87856$.

\subsection{Max-Cut on $3$-Colorable Graphs}			\label{sec:3col-overview}

Now, to establish Theorem \ref{thm:3-col}, we would like to show that same $\alpha_{GW}$-hardness holds for tripartite graphs as well. This naturally motivates us to look at a tripartite variant of noisy-cube gadget, namely the tripartite noisy-cube graph $\cG^R_{\rho^*}(V,E)$, which is defined formally below.

{\bf Vertex Set}. For every $i \in \{1,2,3\}$, we create a copy of the vertex set $V_i := \{(i,x) : x \in \{0,1\}^R\}$. Let $V = \cup_{i \in [3]} V_i$.

{\bf Edge Set }. The distribution on the edges is given by the following stochastic process:

\begin{itemize}
	\item Sample a random pair of indices $(i,j), i \neq j$ from $\{1,2,3\}$.
	\item Sample a $\rho^*$-correlated pair $(x,y) \in \{0,1\}^R \times \{0,1\}^R$ as in Figure \ref{fig:corr}.
	\item Output the edge $\{(i,x),(j,y)\}$.
\end{itemize} 

In other words, in the tripartite noisy-cube, we essentially have $3$ copies of the hypercube vertex-set $\{0,1\}^R$, and we have $\rho^*$-correlated edges going across each of the three pairs $V_i,V_j$. Let us now briefly analyze the properties of this gadget. Clearly, in $\cG^R_{\rho^*}$, each of the sets $V_1,V_2,V_3$ is an independent set, and hence the graph $\cG^R_{\rho^*}$ is $3$-colorable. Furthermore, when this gadget is combined with \ug~problem as the outer-verifier, the final hard instance output by our reduction will also inherit the $3$-colorability of $\cG^R_{\rho^*}$.

Moreover, the completeness of the gadget can be argued as before -- we can define an analogous dictator set $S := \cup_{i \in [3]}\{(i,x): x(1) = 1\}$ -- and using identical arguments as above, one can show that this again cuts $\frac{1 - \rho^*}{2}$-fraction of the edges in $\cG^R_{\rho^*}$. Now, for the soundness guarantee, we would again like to use Majority-Is-Stablest (Theorem \ref{thm:mis}) to conclude that quasirandom sets in $\cG^R_{\rho^*}$ cut at most $\frac{\arccos \rho^*}{\pi}$-fraction of edges. Unfortunately, it turns out that we cannot apply Theorem \ref{thm:mis} as is in our setting, as there are several key differences in the behavior of quasirandom sets between $H^R_{\rho^*}$ and $\cG^R_{\rho^*}$ that our analysis would need to address:

\begin{itemize}
	\item Firstly, note that the partite structure of $\cG^R_{\rho^*}$ can induce degenerate cuts due to which the completeness-soundness guarantees may become trivial for certain values of $\rho$. In particular, $\cG^R_{\rho^*}$ admits degenerate partitions of the form $(V_i,V_j \cup V_k)$ for distinct $i,j,k$, which cut $2/3$ fraction of edges. While these cuts do not by themselves violate the intended soundness, they also do not rule out the possibility that some other variant of these cuts could exploit the partite structure to break the soundness guarantee.
	\item Furthermore, since in this setting one has the freedom to choose different cuts within the three sets $V_1,V_2,V_3$, one cannot assume a priori that an optimal cut is balanced\footnote{Specifically, the soundness analysis for $H^R_{\rho^*}$ in \cite{KKMO07} begins with an argument which shows that it suffices to just work with the {\bf odd part} of the input function indicating the set, which by definition is balanced. Such an argument cannot be carried out in the setting of $\cG^R_{\rho^*}$ due to the fact that one has freedom to choose three different subsets within $V_1,V_2,V_3$.\label{fnt:balanced}}.
\end{itemize}

To address the above issues, we develop a tripartite variant of Majority-Is-Stablest, which is a key technical contribution of this work. We describe the lemma below:

\begin{lemma}[Informal Version of Lemma \ref{lem:main}]				\label{lem:stab-inf}
	Let $\rho \in (-1,0]$. For any quasirandom function $f:V_1 \cup V_2 \cup V_3 \to \{0,1\}$, we have
	\[
	\Pr_{e \sim \cG^R_{\rho^*}}\Big[\text{$f$ cuts $e$}\Big] \leq \frac{\arccos \rho^*}{\pi} + o(1)
	\] 
	Moreover, the RHS is attained (up to $o(1)$-term) by $f(i,x) := {\sf MAJ}(x)$, where ${\sf MAJ}$ is the majority function.
\end{lemma}

We point out that our actual lemma is more general, and completely characterizes the tripartite noisy-cube gadget for all $\rho \in (-1,0]$. 
Establishing the above forms the technical core of our soundness analysis for Theorem \ref{thm:3-col}, and we give a brief outline of its proof in Section \ref{sec:stab-inf}. The above combined with the completeness guarantee of $\cG^R_{\rho^*}$ immediately establishes the properties needed to prove Theorem \ref{thm:3-col}.

\begin{remark}
	It is interesting to compare the above tripartite isoperimetry question to other multipartite variants of Noise-Cube/Gaussian space isoperimetry questions that arise naturally in the context of other graph-cut problems. A particularly relevant problem in this context is the {\em Plurality-Is-Stablest (PIS)} conjecture~\cite{KKMO07,MOO10} where the objective is to characterize the quasirandom $q$-partite cuts that maximize the fraction of edges going across the sets (for e.g., see the Max-$q$-Cut problem~\cite{Frieze-Jerrum}). Making progress on PIS has been quite challenging, and in particular, it has been shown that even the $q = 3$ case behaves quite differently from $q = 2$ (i.e., MIS)~\cite{HMN16}. Hence, it is a bit surprising that the generalization of MIS considered in our paper admits a relatively clean resolution.
\end{remark}

\subsection{Max-Cut with Independent Sets I}

Now we turn our attention to Theorem~\ref{thm:iset-2}. For this reduction, we create a family of gadgets $\{\cG^{\alpha}\}_{\alpha \in [0,1]}$, which are a variation of the tripartite noisy-cube gadget; we describe the gadget informally below. For any $\alpha \in [0,1]$, consider the graph $\cG^{\alpha}$ which is defined as follows:

{\bf Vertex Set}. The vertex set of the graph is $V_1 \cup V_2$ where each $V_i = \{(i,x):x \in \{0,1\}^R\}$ is a copy of the vertex set of the hypercube.

{\bf Edge Weights}. The edges in $\cG^\alpha$ are defined as follows: with probability $\alpha$, we draw a $\rho^*$-correlated pair across $V_1 \times V_2$, and with probability $1 - \alpha$, we draw a $\rho^*$-correlated pair in $V_2 \times V_2$.

Given the above gadget, we can easily verify the following properties:
\begin{itemize}
	\item $V_1$ is an independent set in $\cG^\alpha$. Furthermore, exactly $\alpha$-fraction of edges cross the cut $(V_1,V_2)$.
	\item The dictator set consisting of the strings $\{(i,x): x(1) = 1,i = 1,2\}$ cut $c^* = \frac{1 - \rho^*}{2}$ fraction of edges.
\end{itemize}

As before, establishing the soundness guarantee again requires more work. Following the proof of Lemma \ref{lem:stab-inf} with some appropriate modifications, we are able to show the following guarantee:
\begin{itemize}
	\item If $\alpha \leq s^*$, then quasirandom functions cut at most $s^* + o(1)$-fraction of edges in $\cG^\alpha$.
	\item If $\alpha > s^*$, we end up in the degenerate case where the cut $(V_1,V_2)$ -- which cuts exactly $\alpha$-fraction of the edges -- itself achieves the optimal cut value for quasirandom cuts.
\end{itemize}
Hence by choosing $\alpha = s^*$, we get a gadget which can be used to establish Theorem \ref{thm:iset-2}.

\subsection{Max-Cut with Independent Sets II}               \label{sec:iset-overview}

Finally, let us discuss the key ideas for proving Theorem \ref{thm:i-set}. As a first step, it is useful to take a moment to observe the qualitative differences between Theorem \ref{thm:i-set} and Theorem \ref{thm:iset-2}. Note that in the setting  of Theorem \ref{thm:iset-2}, the setting $\alpha = s^*$ represents a natural limit on the soundness guarantees of the family of instances $\{\cG^\alpha\}_{\alpha}$. This is due to the following reason: since the independent set is known to the algorithm, setting $\alpha > s^*$ immediately breaks the soundness guarantee as the algorithm can just output the cut corresponding to the planted partition $(V_1,V_2)$. Therefore, in order to increase the fraction of edges incident on the independent set while maintaining the same soundness guarantee, we will need a gadget that can ``hide'' the independent set from the algorithm. 

To that end, let us forget the task of proving hardness of Max-Cut for now, and revisit gadgets which can hide independent sets effectively. Such combinatorial objects have been studied in previous works~~\cite{KR08,BK09}, and they are broadly defined in the following way. For any $\alpha \in [0,1]$, consider the following joint distribution $\mu_{\alpha}$ over the Boolean random variables:
\begin{figure}[ht!]
	\begin{center}
		\begin{tabular}{| c | c | c |}
			\hline
			$x$ & $y$ & Probability \\ 
			\hline
			$0$ & $0$ & $1 - \alpha$ \\  
			$0$ & $1$ & $\alpha/2$ \\
			$1$ & $0$ & $\alpha/2$ \\
			$1$ & $1$ & $0$ \\    
			\hline
		\end{tabular}
	\end{center}
	\caption{The Distribution $\mu_{\alpha}$}
	\label{fig:dist}
\end{figure}

Let $\cI^R_{\alpha}$ be the weighted graph on $\{0,1\}^R$, where edges in $\{0,1\}^R \times \{0,1\}^R$ are distributed as the $R$-fold product measure $\mu^{R}_{\alpha}$. Due to the definition of distribution $\mu_{\alpha}$, the graph $\cI^R_{\alpha}$ will satisfy the following properties:
\begin{itemize}
	\item The set $I := \{x \in \{0,1\}^R : x(1) = 1\}$ is an independent set in $\cI^R_{\alpha}$.
	\item Exactly $\alpha$-fraction of edges are incident on $I$ in $\cI^R_{\alpha}$.
\end{itemize}
Furthermore, the spectral profile of $\cI^R_{\alpha}$ ensures that quasirandom functions cannot certify the existence of large independent sets in $\cI^R_{\alpha}$. Due to the above properties, these gadgets have been quite useful for showing hardness results for the {\sc Max Independent Set} problem~\cite{KR08,BK09}. However, since our objective is to show hardness results for Max-Cut, how do we re-purpose the above gadget to fit our requirements?

The answer to this question is: we simply consider the tensor-product of the $\rho^*$-noisy cube $H^R_{\rho^*}$, and the above gadget $\cI^R_{\alpha}$, i.e., $H^R_{\rm fin} = H^R_{\rho^*} \otimes \cI^R_{\alpha}$. Due to the product structure, the completeness properties of the two gadgets are inherited simultaneously by $H^R_{\rm fin}$ i.e., one can immediately show: 
\begin{itemize}
	\item There exists an independent set $I \subset H^R_{\rm fin}$ that is incident on an $\alpha$-fraction of the edges in $H^R_{\rm fin}$.
	\item There exists a dictator cut $S \subset H^R_{\rm fin}$ which cuts $\frac{1 - \rho^*}{2}$ fraction of edges.
\end{itemize} 

However, as one would expect, arguing the soundness guarantee is trickier as our gadget now is a composition of two different combinatorial objects, with potentially conflicting behaviors. Specifically, this results in the following issue in the soundness analysis: if the parameter $\alpha$ is not chosen carefully, one can use the correlation structure of $\cI^R_{\alpha}$ to cheat and certify a large Max-Cut in $H^R_{\rm fin}$ using quasirandom functions. For instance, if we set $\alpha$ to be too large, say $\alpha = 1 - \epsilon$, then we can check that the graph $\cI^R_{\alpha}$ is almost bipartite, and this almost-bipartiteness is inherited by $\cH^R_{\rm fin}$ as well, thus breaking its soundness guarantee. On the other hand, we would like $\alpha$ to be as large as possible since we want to improve on the independent-set size guarantee of Theorem \ref{thm:iset-2}. 

So how do we choose the value of $\alpha$? The key idea here is that we want to choose $\alpha$ in a way such that the correlation structure of $\cI^R_{\alpha}$ doesn't overwhelm the structure of $H^R_{\rho^*}$. In particular, consider the base graph $H_{\rm fin} = H_{\rho^*} \otimes \cI_{\alpha}$ on $\{0,1\}^2$. One can verify that its eigenvalues are $\{1, \rho(\mu_\alpha),\rho^*,\rho(\mu_\alpha)\cdot\rho^*\}$ where $\rho(\mu_{\alpha})$ is the correlation among the variables in $\mu_\alpha$. Now, if we want the largest quasirandom cuts in our gadget to be determined entirely by the structure of $H^R_{\rho^*}$, we need to ensure that the smallest eigenvalue of the base graph $H_{\rm fin}$ is $\rho^*$. This in turn can be ensured by choosing $\alpha$ such that $|\rho(\mu_{\alpha})| \leq |\rho^*|$. Using elementary computations, we can check that $\alpha \approx 0.81597$ is the value\footnote{We point the readers to Corollary \ref{corr:dist} for the precise choice of our parameters and more details regarding this step.} for which $|\rho(\mu_{\alpha})|$ matches $|\rho^*|$, and hence we use this setting of $\alpha$ for our reduction.

Having ensured the required eigenvalue conditions, we now proceed to give a brief outline of the soundness analysis of the reduction. Firstly, by analyzing the correlation structure of the test distribution, we first argue that quasirandom functions can only use the correlation structure of $H^R_{\rho^*}$ to certify the largest cuts. We then use Mossel's Majority-Is-Stablest theorem for general domains~\cite{Mossel10} to bound the soundness value of the gadget. While this step largely follows \cite{KKMO07}'s soundness analysis, there are still a few differences which we highlight below:

\begin{itemize}
	\item Since our reduction needs to maintain perfect completeness, we cannot introduce additional noise in the test to dampen the high-degree terms in the test. Instead, this is handled by the observation that since our test distribution is connected, we can introduce {\em implicit noise} in the soundness analysis. 
	
	\item Furthermore, in our dictatorship test, as the long code tables are indexed over $(\{0,1\}\times\{0,1\})^R$ as opposed to $\{0,1\}^R$, a priori, we cannot begin the soundness analysis with the assumption that the long code table is balanced. This issue is instead handled using a {\em convexity argument} after Majority-Is-Stablest is applied.
\end{itemize}

\subsection{Overview of the Algorithm}

We now briefly discuss the key ideas underlying our algorithm for Theorem \ref{thm:algo}. Recall that in this setting, we are given a graph $G = (V,E)$ such that it contains an independent set $I \subseteq V$ satisfying $|E[I,I^c]| \geq \alpha|E|$ where $\alpha = \alpha^* + \epsilon$ for some constant $\epsilon > 0$. Now, note that the cut $|E[I,I^c]|$ itself is quite large, so if we could output that cut, or some close approximation to it, we might be able to get better than $\alpha_{GW}$-approximation. For this, our first idea is to compete with respect to the ``planted cut'' $(I,I^c)$, and ignore all other maximum cuts altogether. 

Now, how do we find a (good enough approximation) to the planted cut\footnote{We call the cut which has the large independent set on one side as the planted cut.}? To that end, let us recall the $\{0,1\}$-SDP relaxation for Max-Cut:
\[
\begin{array}{ll}
	\text{max} & \frac{1}{|E|}\sum_{\{i, j\} \in E} \|u_i - u_j\|^2_2 \\
	\text{s.t.} & \forall i \in V, 0 \leq \|u_i\|^2_2 \leq 1\\
\end{array}
\]
Now, since we want to target the cut corresponding to $(I,I^c)$, we want to further enforce the constraint that any feasible integral solution to the above program must correspond to an independent set. In SDP relaxations, the standard way to enforce this is to add the orthogonality constraint $\langle u_i,u_j \rangle = 0$, for every edge $\{i,j\} \in E$. Subject to this constraint, we see that for any edge $\{i,j\} \in E$, the corresponding term in the objective simplifies to $\|u_i - u_j\|^2_2 = \|u_i\|^2_2 + \|u_j\|^2_2$. As is standard, we also add a $0$-vector $u_0$ to align the SDP relaxation. This results in the following SDP relaxation:
\[
\begin{array}{ll}
	\text{max} & \frac{1}{|E|}\sum_{\{i, j\} \in E} \|u_i\|^2 + \|u_j\|^2_2 \\
	\text{s.t.} & \forall i \in V, 0 \leq \|u_i\|^2_2 \leq 1\\
				& \forall i \in V, \langle u_i,u_0 \rangle = \|u_i\|^2_2 \\
				& \forall \{i,j\} \in E, \langle u_i,u_j \rangle = 0
\end{array}
\]
Overall, this final relaxation can be thought of as a combination of the $\{0,1\}$-version of the Goemans-Williamson relaxation, and the Lovasz-Theta relaxation for finding orthogonal embeddings of graphs. 

{\bf Rounding}. Now, as a next step, we solve the above SDP relaxation, at which point we are left with a vector SDP solution, which still needs to be rounded to a good integral solution. To that end, we employ the $\THRESH^-$ family of rounding schemes, proposed by Lewin, Livnat, and Zwick~\cite{lewin2002improved}. In particular, for a threshold function $t:[0,1]\to[0,1]$, the $\THRESH^-$ rounding scheme is defined as follows. For every $i \in V$, let $u^\perp_i$ denote the (normalized) component of $u_i$ that is orthogonal to $u_0$, and let $\mu_i = \langle u_i,u_0 \rangle$. Then, we do the following: 

\begin{itemize}
	\item Sample a random Gaussian vector $g \sim N(0,1)^n$.
	\item Let $S := \{i \in V : \langle g,u^\perp_i \rangle \leq \Phi^{-1}(t(\mu_i))\}$.
	\item Output the cuts $(S,S^c)$. 
\end{itemize}

While this family of schemes is incredibly versatile, it comes with the complication that for most non-trivial choices of the threshold function $t(\cdot)$, the approximation guarantee doesn't admit a closed form solution, and rigorous verification of its guarantees can be quite complicated.
To that end, most of the work goes into designing the ``right'' threshold function for which we can achieve the desired guarantee.

{\bf Design of Threshold Function}. In order to design our threshold function, our first key observation is to observe that given an edge $(i,j)$ such that $\mu_i = \mu_j = \alpha/2$, our rounding function must simulate hyperplane rounding on that edge i.e., $t(\alpha/2) = 1/2$. Furthermore, assuming certain natural first-order conditions on the soundness-curve of this rounding scheme, we are able to deduce certain additional controls points for the threshold function. Our final choice of the threshold function is a piece-wise linear curve that interpolates the values at these control points (see Table \ref{table:t}).

{\bf Analysis}. Given the above rounding scheme, we then proceed to give a rigorous analysis of its approximation guarantee. This is done by  identifying carefully chosen certificates that (a) imply better than $\alpha_{GW}$-approximation guarantee of the rounding scheme and (b) can be verified using interval arithmetic. We identify these certificates using a careful analysis of the derivatives the soundness-curve. We then verify these certificates using interval arithmetic, which in turn establishes the approximation guarantees claimed in Theorem \ref{thm:algo}.

\subsection{Proof Sketch of Lemma \ref{lem:stab-inf}}		\label{sec:stab-inf}

Recall that in the context of this lemma our goal is to characterize the quasirandom functions $f:V_1 \cup V_2 \cup V_3 \to \{0,1\}$ which cut the maximum number of edges in the tripartite noise-graph $\cG^R_{\rho^*}$. For ease of notation, throughout this section we will denote $\rho := \rho^*$. Our proof proceeds along the following sequence of steps outlined below.

{\bf Step 1: Arithmetization}. Let us fix such an $f$, and for every $i \in \{1,2,3\}$, let $f_i:V_i \to \{0,1\}^R$ be the cut corresponding to the $i^{th}$-partition. Now, the edges of $\cG^R_{\rho}$ are distributed as a $\rho$-correlated pair of strings $(x,y)$ that goes across $(V_i,V_j)$ for a random pair $(i,j)$. Hence, we can express the fraction of edges cut by $f$ in $\cG^R_{\rho}$ as:
\begin{equation}			\label{eqn:e1}
\Pr_{e \sim \cG^R_{\rho}}\Big[\textit{$e$ is cut}\Big]
= \Ex_{(i,j) \sim [3]}\Ex_{x \underset{\rho}{\sim} y}\Big[{\sf Cut}(f_i(x),f_j(y))\Big].
\end{equation}
where ${\sf Cut}(a,b) := a + b - 2ab$ is the arithmetization of the cut-function over $\{0,1\}$, and $(i,j) \sim [3]$ denotes the draw of a random pair of indices from $\{1,2,3\}$.

{\bf Step 2: Moving to Gaussian Space}. Since $f$ is quasirandom, the individual sets $f_1,f_2,f_3$ are quasirandom as well. Hence, using the {\it two-function Gaussian stability bound}~\cite{DRM06,MOO10}, we can upper bound the RHS of \eqref{eqn:e1} as:
\begin{equation}			\label{eqn:e2}
	\Ex_{(i,j) \sim [3]}\Ex_{x \underset{\rho}{\sim} y}\Big[{\sf Cut}(f_i(x),f_j(y))\Big]
	\leq \Ex_{(i,j) \sim [3]}\Ex_{x \underset{\rho}{\sim} y}\Big[{\sf Cut}_\rho(\mu_i,\mu_j)\Big]  + o(1),
\end{equation}	
where $\mu_i = \Ex[f_i]$ is the volume of the cut within the $i^{th}$ partition, and ${\sf Cut}_\rho(\cdot,\cdot)$ denotes the $\rho$-correlated Gaussian cut function\footnote{Formally, ${\sf Cut}_\rho(a,b) = \Pr_{g_1 \underset{\rho}{\sim} g_2}\left[h_{a}(g_1) \neq h_{b}(g_2)\right]$ where $h_a,h_b$ are halfspaces of Gaussian measure $a$ and $b$ respectively.}. Hence, we can shift our objective to that of characterizing the maximizers of the following optimization problem in the Gaussian space:
\begin{equation}			\label{eqn:e3}
	\max_{(\mu_1,\mu_2,\mu_3) \in [0,1]^3} \Ex_{(i,j) \sim [3]} \Big[{\sf Cut}_\rho(\mu_i,\mu_j)\Big],
\end{equation}

In particular, the above corresponds to the following tripartite isoperimetric problem in the Gaussian space: consider the tripartite graph $\cQ_\rho$ on $(\mathbb{R})_1 \cup (\mathbb{R})_2 \cup (\mathbb{R})_3$, where edges correspond to a $\rho$-correlated Gaussians on $(\mathbb{R})_i \times (\mathbb{R})_j$ for a random pair of indices $(i,j)$. In this graph, we would like to identify halfspaces\footnote{Again, here $h_\mu$ denotes the halfspace on $\mathbb{R}$ of Gaussian measure $\mu$.} $h_{\mu_1},h_{\mu_2},h_{\mu_3}$ such that the set $(h_{\mu_1} \cup h_{\mu_2} \cup h_{\mu_3})$ cuts the maximum weight of edges in $\cQ_\rho$. Identifying these optimal thresholds $\mu_1,\mu_2,\mu_3$ forms the technical core of our proof, and we briefly outline the argument below.

Let $\bmu^* := (\mu^*_1,\mu^*_2,\mu^*_3)$ be an optimizer for the above problem. We will proceed by considering two cases, based on whether $\mu^*$ lies in the interior of $[0,1]^3$, or its surface.

{\bf Case 1: $\bmu^*$ is in interior}. For brevity,  for the rest of the description we will denote  $F(\mu_1,\mu_2,\mu_3) = \Ex_{(i,j) \sim [3]}[{\sf Cut}_\rho(\mu_i,\mu_j)]$. Since $\bmu^*$ is in the interior, that implies that all the partial derivatives of $F$ must evaluate to $0$ at $\bmu^*$. Next, by manipulating the partial derivative expressions, we observe that we can write the partial derivative w.r.t. $\mu_1$ (and similarly, w.r.t. $\mu_2,\mu_3$) as:
\[
\frac{\partial F}{\partial \mu_1} = G(K(\mu_1,\mu_2)) + G(K(\mu_1,\mu_3)),
\]
for some functions $G$ and $K$ which satisfy the following:
\begin{itemize}
	\item $G:\mathbb{R} \to \mathbb{R}$ is an odd and strictly increasing function.
	\item $K:\mathbb{R}^2 \to \mathbb{R}$ is defined as\footnote{$\Phi^{-1}(\cdot)$ is the inverse of the Gaussian CDF function.}:
	\[
		K(x,y) := \frac{\Phi^{-1}(y) - \rho\cdot\Phi^{-1}(x)}{\sqrt{1 - \rho^2}}.
	\]
\end{itemize}
Now, by using the first order optimality conditions $\frac{\partial F}{\partial \mu_i} = 0$ for $i = 1,2,3$, and using the properties of the function $G$, we can further deduce that: 
\[
K(\mu_i,\mu_j) + K(\mu_j,\mu_k) = 0,	\qquad\qquad \forall i \neq j \neq k \in \{1,2,3\}.
\]
The above, along with the definition of $K(\cdot,\cdot)$ yields the following set of linear equations:
\begin{equation}			
	\begin{bmatrix}
		-2 \rho & 1 & 1 \\
		1 & -2\rho & 1 \\
		1 & 1 & - 2 \rho 
	\end{bmatrix}
	\begin{bmatrix}
		\Phi^{-1}(\mu_1)  \\  \Phi^{-1}(\mu_2)  \\ \Phi^{-1}(\mu_3)
	\end{bmatrix}
	= 
	\begin{bmatrix}
		0  \\  0  \\ 0
	\end{bmatrix}
	.
\end{equation}
For $\rho \in (-1,0) \setminus\{-1/2\}$, the above coefficient matrix is non-singular, and hence the only solution to the above system corresponds to $\Phi^{-1}(\mu_i) = 0$ for all $i = 1,2,3$, which in turn implies that $\mu_1 = \mu_2 = \mu_3 = 1/2$. In this case, the standard computation yields that the corresponding optimal value is 
\[
F(\bmu^*) = {\sf Cut}_\rho(1/2,1/2) = \frac{\arccos \rho}{\pi}.
\]

{\bf Case (ii): $\bmu^*$ is on the surface}. In this case, following a similar and relatively simpler argument we can show that all the optimizers of $F$ on the boundary must be (up to permutations) of the form $(0,1,x)$ for any $x \in (0,1)$, i.e., these triples correspond to the degenerate cuts where two of the partitions are placed on the opposite sides, and the partition within the third set can be chosen arbitrarily. Furthermore, for any such degenerate cut, the optimal value is $F(\bmu^*) = 2/3$.

{\bf Putting Things Together}.  The proof is concluded by observing that for $\rho = \rho^* = -0.689$, the maximizer from the interior, i.e., $(1/2,1/2,1/2)$, yields a larger value of the objective, which is $(\arccos \rho^*)/\pi$, which concludes the proof.

\begin{remark}[Proof for general $-1 < \rho \leq 0$]
	Our proof of Lemma \ref{lem:stab-inf} for general values of $\rho$ proceeds exactly as above; the only difference is that in the end we identify that for all $\rho < -1/2$, the balanced cuts $(1/2,1/2,1/2)$ from Case (i) yields a larger value, whereas for $\rho \in [-1/2,0]$, the degenerate cuts from Case (ii) yield a larger cut-value.
\end{remark}

\section{Preliminaries}

We introduce some basic notation and concepts that will be used in the paper. For an integer $i$, we will use $[i]$ to denote the set $\{1,2,\ldots,i\}$. For a finite set $\Omega$, and a vector $x \in \Omega^R$, we will use the notation $x(i)$ to denote the $i^{th}$ entry of $x$. We will also define the action of a permutation on a vector. Given a vector $x \in \Omega^R$ and a permutation $\pi:[R] \to [R]$, we use $\pi(x)$ to denote the vector $(x(\pi(1)),x(\pi(2)),\ldots,x(\pi(R)))$.

{\bf Weighted Graphs}. We will often work with weighted graphs, for which we setup some conventions. Throughout, given a weighted graph $G = (V,E,w)$, we will assume that the weight function $w:E \to \mathbbm{R}_{\geq 0}$ is non-negative and it sums to $1$ i.e., $w$ induces a distribution on the set of edges $E$. Hence, drawing an edge $e \sim G$ corresponds to sampling an edge from the distribution induced by the weight function. With this convention, for any cut $(S,S^c)$ of $G$, the quantity $w(S,S^c)$ represents the weight of the edges crossing the cut. Furthermore, if $G$ is a Max-Cut instance, then we use ${\sf opt}(G)$ to denote the maximum weight of a cut in $G$ i.e.,
\[
{\sf opt}(G) := \max_{S \subseteq V} w(S,S^c).
\]

\subsection{Fourier Analysis}

We review some elementary Fourier analytic concepts and tools from \cite{OD14} that will be used in the analysis of our reductions. Given a finite set $\Omega$, and a measure $\mu$ on $\Omega$, we refer to the pair $(\Omega;\mu)$ as a finite probability space. We use $L_2(\Omega;\mu)$ to denote the space of real-valued functions $f:\Omega \to \mathbbm{R}$ that are square-integrable with respect to measure $\mu$. The space $L_2(\Omega;\mu)$ is endowed with the natural inner product
\[
\langle f, g \rangle \defeq \Ex_{\omega \sim \mu}\Big[f(\omega)g(\omega)\Big]
\]
It is well-known that $L_2(\Omega; \mu)$ admits an orthonormal basis $\{\phi_0 \equiv 1, \phi_1,\ldots,\phi_{\ell - 1}\}$ with $\ell:= |\Omega|$, known as the {\em Fourier basis}. Using this basis, one can express any function $f \in L_2(\Omega;\mu)$ as
\[
f = \sum_{i = 0}^{\ell - 1} \wh{f}(i)\phi_i,
\]
where $\{\wh{f}(i)\}_i$ are called the Fourier coefficients of $f$. 

{\bf Product Spaces}. Given a probability space $(\Omega;\mu)$, for any $R > 1$, one can consider its $R$-fold product $(\Omega^R;\mu^R)$. It is easy to check that then the set of functions $\{\phi_\beta\}_{\beta \in \mathbbm{Z}^R_{\ell}}$ form a Fourier basis for $L_2(\Omega^R;\mu^R)$, where for any $\beta$, we have $\phi_\beta(x) := \prod_{i \in [R]} \phi_{\beta(i)}(x(i))$. Consequently, any function $f \in L_2(\Omega^R;\mu^R)$ admits the Fourier expansion:
\[
f = \sum_{\beta \in \mathbbm{Z}^R_{q}} \wh{f}(\beta)\phi_{\beta}, 
\]
where $\{\wh{f}(\beta)\}_{\beta}$ are the Fourier coefficients of $f$. We recall a couple of useful Fourier analytic facts.

\begin{fact}
	For any function $f \in L_2(\Omega^R;\mu^R)$,
	\begin{itemize}
		\item $\wh{f}(0) = \Ex_{x \sim \mu^R}[f(x)]$.
		\item {\bf Parseval's Identity}.
		\[
		\sum_{\beta \in \mathbbm{Z}^R_{q}}\wh{f}(\beta)^2 = \|f\|^2_2,
		\]
		where $\|\cdot\|^2_2$ norm is defined in $L_2(\Omega^R;\mu^R)$. 
	\end{itemize}
\end{fact}

{\bf Influence}. For a function $f \in L_2(\Omega^R;\mu^R)$, the influence of its $i^{th}$ coordinate, denoted by $\Inf{i}{f}$ is its expected conditional variance along the $i^{th}$-variable, i.e.,
\[
\Inf{i}{f} = \Ex_{x_{\neq i} \sim \{0,1\}^{[R]\setminus \{i\}}}\left[{\sf Var}_{x_i}\Big[f(x_i,x_{\neq i})\Big]\right].
\]

{\bf Correlated Pairs and Noise Operator}. For a vector $x \in \{0,1\}^R$, and a parameter $\rho \in [-1,1]$, a $\rho$-correlated copy of $x$ is a random vector $y \sim \{0,1\}^R$ which is drawn from the following distribution. For every $i \in [R]$, we sample the bit $y(i)$ independently as follows:
\[
y(i) = 
\begin{cases}
	x(i)	& \text{ with probability } \frac{1 + \rho}{2},	\\
	1 - x(i)	& \text{ with probability } \frac{1 - \rho}{2}
\end{cases}
\]
We use the notation $y \underset{\rho}{\sim} x$ denote that $y$ is a $\rho$-correlated copy of $x$. For any $\rho \in [-1,1]$, we let $T_\rho$ denote the {\em $\rho$-correlated noise operator}, which is a stochastic operator on $L_2(\{0,1\}^R)$ defined as follows. For any function $f:\{0,1\}^R \to \mathbbm{R}$, the function $T_\rho f:\{0,1\}^R \to \mathbbm{R}$ is defined as
\[
T_\rho f(x) := \Ex_{y \underset{\rho}{\sim} x}\big[f(y)\big].
\]

\subsection{Gaussian Noise Stability } \label{sec:corr-def}

For $\rho \in [-1,1]$, a pair of Gaussian random variables $(g_1,g_2)$ is said to be {\bf $\rho$-correlated} if 
\[
(g_1,g_2) \sim N\left(
\begin{bmatrix}
	0 \\ 0
\end{bmatrix},
\begin{bmatrix}
	1 & \rho \\
	\rho & 1 
\end{bmatrix}
\right)
\] 
We will frequently use the notation $g_1 \underset{\rho}{\sim} g_2$ to denote that $(g_1,g_2)$ are $\rho$-correlated. 

Now, for $\rho \in [-1,1]$ and  $\mu_1,\mu_2 \in [0,1]$, we use $\Lambda_\rho(\mu_1,\mu_2)$ to denote the Gaussian quadrant probability function, which is defined as
\begin{equation}        \label{eqn:def:lambda-rho}
\Lambda_\rho(\mu_1,\mu_2) := \Pr_{g_1 \underset{\rho}{\sim} g_2}\Big[g_1 \leq \Phi^{-1}(\mu_1), g_2 \leq \Phi^{-1}(\mu_2)\Big],
\end{equation}
where $\Phi^{-1}(\cdot)$ denotes the inverse Gaussian CDF function. The following properties of $\Lambda_\rho$ are well-known.

\begin{fact}[\cite{austrin-sat-full}]	\label{fact:convex}
	For any $\rho \in (-1,1)$, the map $x \mapsto \Lambda_\rho(x,x)$ is convex.
\end{fact}

One proof of the above fact can be found below Corollary D.6 in \cite{austrin-sat-full}.

\begin{lemma}[Sheppard's Formula]      \label{lem:shep}
	For any $\rho \in [-1,1]$ we have
	\begin{align*}
    \Pr_{g_1 \underset{\rho}{\sim} g_2}\Big[{\rm sign}(g_1) \neq {\rm sign}(g_2)]  = \frac{\arccos \rho}{\pi}.
	\end{align*}
    Consequently, we also have:
    \[
    1 - 2\Lambda_{\rho}\left(\frac12,\frac12\right) = \frac{\arccos \rho}{\pi}.
    \]
\end{lemma}

\subsection{Correlation and Noise Stability}		

Let $(\Omega^2; \mu)$ denote a finite probability space corresponding to a joint distribution on a pair of variables, where the variables take values from $\Omega$, and their joint distribution is specified using the measure $\mu$ on $\Omega^2$. Then the correlation of $(\Omega^2; \mu)$ is defined as
\[
\rho(\Omega^2; \mu) := \max_{\substack{f,g:\Omega \to \mathbbm{R} \\ \Ex f = \Ex g = 0 \\ \|f\|_2,\|g\|_2 \leq 1}} \Ex_{(x,y) \sim \mu}\big[f(x)g(y)\big].
\] 

{\bf Bounds for Low-Influence functions}. For the soundness analysis, we shall rely on the following noise stability bound.
\begin{theorem}[Theorem 1.14~\cite{Mossel10}]			\label{thm:stab}
	Let $(\Omega \times \Omega; \mu)$ be a finite probability space corresponding to a joint distribution over $\Omega \times \Omega$. For $i \in \{1,2\}$, let $f_i \in L_2(\Omega^R; \mu_i^R)$ be a function satisfying $\|f_i\|_\infty \leq 1$, where $\mu_i$ is the margin of $\mu$ on the $i$-th coordinate. Furthermore, suppose 
	\[
	\max_{a \in \{1,2\}}\max_{j \in [R]} \Inf{j}{f_a} \leq \tau.
	\]
	Let $\rho := \rho(\Omega^2; \mu)$. Then,
	\begin{equation}			\label{eqn:bound}
	\Ex_{(x,y) \sim \mu^{R}}\left[f_1(x)f_2(y)\right] \geq \Lambda_{-|\rho|}(\Ex f_1, \Ex f_2) - c(\tau).
	\end{equation}
    where $c(\tau)$ is an increasing function of $\tau$ satisfying $c(0) = 0$.
\end{theorem}

We point out that Theorem 1.14 in \cite{Mossel10} doesn't exactly give the bound in the form stated above in \eqref{eqn:bound}; however it is an immediate corollary of the bound stated in \cite{Mossel10}. We include an explanation of how Theorem 1.14 in \cite{Mossel10} implies \eqref{eqn:bound} in Appendix \ref{sec:thm-stab}. 

{\bf Adding Implicit Noise}. Some of our dictatorship tests have to satisfy perfect completeness; and hence for our soundness analysis, we will need to introduce implicit-noise. Specifically, we shall need the following lemma from \cite{Mossel10}.

\begin{lemma}		\label{lem:noise}
	Let $(\Omega_1 \times \Omega_2 \times \cdots \times \Omega_k;\mu)$ be a probability space with $\rho(\prod^k_{i = 1}\Omega_i;\mu) = \rho < 1$. Consider the product space $((\Omega_1)^R \times \cdots \times (\Omega_k)^R, \mu^{\otimes R})$. Let $f_1,\ldots,f_k$ be functions where $f_i:\Omega^R_i \to \mathbbm{R}$ are functions satisfying ${\rm Var}[f_i] \leq 1$ for every $i \in [k]$. Then for every $ \epsilon > 0$, there exists $\gamma = \gamma(\rho,\epsilon)$ such that
	\[
	\left|\Ex_{(x_1,\ldots,x_k) \sim \mu^R} \left[\prod_{i \in [k]} f_i(x_i)\right]
	- \Ex\left[\prod_{i \in [k]} T_{1 - \gamma}f_i(x_i)\right]\right| \leq k \epsilon.
	\]
\end{lemma}

\subsection{The Unique Games Conjecture}           \label{sec:ugc}

All our hardness results are based on reduction from the \ug~problem, which is defined below.

\begin{definition}[Unique Games]
	A \ug~instance $\cG(U,V,[R],E,\{\pi_{v,u}\}_{(u,v) \in E})$ is characterized by a left-regular\footnote{Without loss of generality, we can work with left-regular instances of \ug, for e.g., see Lemma 3.3.~\cite{KR08}.} bipartite graph with left and right vertex sets $U$ and $V$, and an edge set $E \subseteq U \times V$. A labeling of the vertices in $\cG$ is a map $\sigma:U \cup V \to [R]$ which assigns the vertices of $\cG$ values from the label set $[R]$. Each edge $(u,v) \in E$ is associated with a {\bf bijection} constraint $\pi_{v,u}:[R] \to [R]$. 
	
	A labeling  $\sigma: U \cup V \to [R]$ is said to satisfy edge $(u,v)$ if $\pi_{v,u}(\sigma(v)) = \sigma(u)$. The value of the \ug~instance $\cG$, denoted by ${\sf val}(\cG)$, is simply the maximum fraction of constraints in $\cG$, that can be satisfied by any labeling $\sigma:U \cup V \to [R]$.   
	
	For any $c,s \in (0,1)$ such that $c \geq s$, the $(c,s)$-\ug~problem is a decision problem where given a \ug~instance $\cG$, the objective is to distinguish between:
	\begin{itemize}
		\item {\bf Yes Case}. ${\sf val}(G) \geq c$.
		\item {\bf No Case}. ${\sf val}(G) \leq s$.
	\end{itemize}
\end{definition}

We now state Khot's Unique Games Conjecture, which is the starting point for all our reductions.

\begin{conjecture}[The Unique Games Conjecture~\cite{Khot02a}]         \label{conj:ugc}
		There exist constants $\epsilon_0,\delta_0 \in (0,1)$ such that the following holds. For every $\epsilon \in (0,\epsilon_0], \delta \in (0,\delta_0]$ there exists $R = R(\epsilon,\delta)$ such that the $(1 - \epsilon,\delta)$-\ug~problem on instances with label set $[R]$ is \NP-hard.
\end{conjecture}

\subsection{Miscellaneous}

We record the following useful observation.

\begin{lemma}				\label{lem:eigval}
	Let $\mu$ be a joint distribution over $\{0,1\}$-valued random variables $X,Y$ such that $\Ex[X] = \Ex[Y] = p$. Let $\{0,1\}_{p}$ denote the $p$-biased measure on $\{0,1\}$. Let $\rho$ denote the correlation between variables $X$ and $Y$. Let $(\phi_0 \equiv 1,\phi_1)$ denote the Fourier basis for $L_2(\{0,1\}_p)$. Then,
	\[
	\Ex_{(X,Y) \sim \mu}\Big[\phi_1(X)\phi_1(Y)\Big] = \rho.
	\]
\end{lemma}
\begin{proof}
	It is known that (e.g., see Definition 8.39~\cite{OD14}) the non-constant basis function can be expressed as $\phi_1(X) := (X-p)/\sqrt{p(1-p)}$. Hence, 
	\[ 
        \Ex_{(X,Y)}\left[\phi_1(X)\phi_1(Y)\right] = \frac{\Ex_{(X,Y)}[XY] - p^2}{p(1 - p)} = \rho. 
	\]
\end{proof}

\section{Hardness Reduction for Theorem \ref{thm:3-col}}

We describe the reduction from the \ug~problem in the figure below:

\begin{figure}[ht!]
	\begin{mdframed}
		Let $\cG = (U,V,E,[R],\{\pi_e\}_{e \in E})$ be a $(1 - \epsilon,\delta)$-\ug~instance as in Conjecture \ref{conj:ugc}, where $\epsilon,\delta > 0$ are chosen below, and $R = R(\epsilon,\delta)$. Let $\eta > 0$ be a constant chosen below. Our hard instance of Max-Cut is going to be a weighted graph $\cH = (V_\cH,E_{\cH},w)$ which we construct as follows. \\
		
		{\bf Vertices}. For every $v \in V$ and $i \in \{1,2,3\}$, we introduce an $R$-dimensional cube $V_{\cH}(v,i) = \{(v,i,x) : x \in \{0,1\}^R\}$. We let 
		\[
		V_{\cH} := \bigcup_{v \in V, i \in [3]} V_{\cH}(v,i),
		\]
		be the entire set of vertices in $\cH$. \\
		
		{\bf Edge Weights}. The weights on the edges are given by the following process: \\[-5pt]
		
		\begin{enumerate}
			\item Sample a uniformly random pair $\{i,j\}$ from $\{1,2,3\}$.
			\item Sample $u \sim U$ and $v_i,v_j \sim N_\cG(u)$. 
			\item Sample $(x,y)$ by first sampling $x \sim \{0,1\}^R$ and then sampling $y \underset{\rho^*}{\sim} x$ as a $\rho^*$-correlated copy of $x$.
			\item Sample $x' \underset{1 - \eta}{\sim} x$ and $y' \underset{1 - \eta}{\sim} y$ independently.
			\item Output\footnote{{Recall that for a string $x \in \{0,1\}^R$ and a permutation $\pi:[R] \to [R]$, the notation $\pi(x)$ denotes the string $(x(\pi(1)),x(\pi(2)),\ldots,x(\pi(R)))$.}} the edge $\{(v_i,i,\pi_{v_i,u}(x')),(v_j,j,\pi_{v_j,u}(y'))\}$. 
		\end{enumerate}
	\end{mdframed}
	\caption{Reduction for Theorem \ref{thm:3-col}}
	\label{fig:3-col}
\end{figure}

{\bf Parameters of the Reduction}. We assign the following values to the parameters. 
\begin{itemize}
\item Let $\eta$ be as in Theorem \ref{thm:3-col}.
\item Let $\tau$ be such that $c(\tau),c_1(\tau) \leq \eta/10$, where $c(\tau),c_1(\tau)$ are as in Theorem \ref{thm:stab} and Lemma \ref{lem:noise} respectively. 
\item Let $\epsilon_0,\delta_0$ be the constants from Conjecture \ref{conj:ugc}. We set $\epsilon = \min\{\eta,\epsilon_0\}$ and $\delta \in (0,\delta_0)$ is set such that $\eta'(\delta,\tau,\eta) \leq \eta$ (in Lemma \ref{lem:inf-dec}).
\end{itemize}

\begin{claim}		\label{cl:3-col}
	The graph $\cH$ output by the reduction in Figure \ref{fig:3-col} is $3$-Colorable.
\end{claim}
\begin{proof}
	Note that the sets $V_\cH(i) := \cup_{v}V_\cH(v,i)$ for $i = 1,2,3$ form a $3$-partition of $V_{\cH}$, and for every $i$,  $V_{\cH}(i)$ is an independent set. Hence the claim follows.
\end{proof}

The following lemmas state the completeness and soundness properties of the above reduction.

\begin{lemma}		\label{lem:comp-3col}
	If ${\sf val}(\cG) \geq 1 - \epsilon$. Then there exists a set $S \subseteq V_\cH$ which cuts at least $\frac{1 - \rho^*}{2} - 4\eta$-fraction of edges in $\cH$.
\end{lemma}

\begin{lemma}		\label{lem:sound-3col}
	Suppose ${\sf val}(\cG) \leq \delta$. Then, every set $S \subseteq V_{\cH}$ cuts at most $\frac{\arccos \rho^*}{\pi} + 4\eta$-fraction of edges in $\cH$.
\end{lemma}

Before we prove the above lemmas, let us use them to conclude the proof of Theorem \ref{thm:3-col}.

\begin{proof}[Proof of Theorem \ref{thm:3-col}]
	Let $\cG$ be a $(1 - \epsilon,\delta)$-\ug~instance, where $\epsilon,\delta$ are chosen as above. Then using the reduction in Figure \ref{fig:3-col}, we can construct a Max-Cut instance $\cH = (V_\cH,E_\cH,w_{\cH})$ in polynomial time such that:
	\begin{itemize}
		\item The graph $\cH$ is $3$-colorable (Claim \ref{cl:3-col}).
		\item If ${\sf val}(\cG) \geq 1 - \epsilon$, then the optimal Max-Cut in $\cH$ has value at least $c^* - 4\eta$ (Lemma \ref{lem:comp-3col}).
		\item If ${\sf val}(\cG) \leq \delta$, then the optimal Max-Cut in $\cH$ has value at most $s^* + 4\eta$ (Lemma \ref{lem:sound-3col}).
	\end{itemize}
	Since from Conjecture \ref{conj:ugc}, the $(1 - \epsilon,\delta)$-\ug~problem is \NP-hard, this implies that it is \NP-hard to distinguish between whether ${\sf opt}(\cH) \geq c^* - 4\eta$ or ${\sf opt}(\cH) \leq s^* + 4\eta$, which proves Theorem \ref{thm:3-col}.
\end{proof}

\subsection{Proof of Lemma \ref{lem:comp-3col}}			

Suppose $\cG$ admits a labeling $\sigma: U \cup V \to [R]$ which satisfies at least $(1 - \epsilon)$-fraction of constraints. We construct a cut $(S,S^c)$ of $\cH$ as follows. For every $v \in V$ and $i \in \{1,2,3\}$, let us define
\[
S_{v,i} := \{(v,i,x): x(\sigma(v)) = 1\},
\]
and let $S := \cup_{v \in V}\cup_{i \in [3]} S_{v,i}$. We now bound the fraction of edges cut by $(S,S^c)$. To that end, let us fix a choice of $u,v_i,v_j$ such that $\sigma$ satisfies both edges $(u,v_i)$ and $(u,v_j)$. Then the fraction of edges cut conditioned on the choice of $u,v_i,v_j$ can be bounded as:
\begin{align*}
	&\Ex_{\{i,j\}}\Pr_{(x',y')}\Big[\one_{S}(v_i,i,\pi_{v_i,u}(x')) \neq \one_S(v_j,j,\pi_{v_j,u}(y'))\Big] \\
	&= \Ex_{\{i,j\}}\Pr_{(x',y')}\Big[\big(\pi_{v_i,u}\circ x'\big)(\sigma(v_i)) \neq \big(\pi_{v_j,u} \circ y'\big)(\sigma(v_j))\Big] \\
	&= \Ex_{\{i,j\}}\Pr_{(x',y')}\Big[x'(\sigma(u)) \neq y'(\sigma(u))\Big]		\tag{Since $\pi_{v_i,u}(\sigma(v_i)) = \pi_{v_j,u}(\sigma(v_j)) = \sigma(u)$} \\
	&\geq \Ex_{\{i,j\}}\Pr_{(x',y')}\Big[x(\sigma(u)) \neq y(\sigma(u))\Big] - 2\eta \\
	&\geq \frac{1 - \rho^*}{2} - 2\eta.
\end{align*}
I.e., conditioned on such a choice of $(u,v_i,v_j)$, the set $S$ cuts at least $\frac{1 - \rho^*}{2} - 2\eta$ fraction of edges. Furthermore, since $\cG$ is left-regular, by a union bound argument it follows that the event ``$\sigma$ satisfies the edge $(u,v_i),(u,v_j)$'' occurs with probability at least $1 - 2\epsilon$ independent of the choice of $i,j,x',y'$. Hence, we can conclude that the fraction of edges cut by $S$ is at least 
\[
\frac{1 - \rho^*}{2} - 2\eta - 2\epsilon.
\]

\subsection{Proof of Lemma \ref{lem:sound-3col}}

Suppose the optimal value of $\cG$ is at most $\delta$. Let $f:V_\cH \to \{0,1\}$ be the indicator corresponding to a subset of $V_\cH$. For every choice $v \in V$ and $i \in \{1,2,3\}$, we will find it convenient to define the function $f_{v,i}:\{0,1\}^R \to \{0,1\}$ as $f_{v,i}(x) = f(v,i,x)$.
Let us also define the function ${\sf Cut}:[0,1]^2 \to [0,1]$ as
\[
{\sf Cut}(x,y) := 1 - xy - (1 - x)(1 - y).
\]
In particular, note that if $x,y$ are Boolean-valued then ${\sf Cut}(x,y) = \one(x \neq y)$. We proceed by arithmetizing the fraction of edges cut by $S$ as 
\begin{align}
	\Pr_{e \sim E_{\cH}}\Big[\text{e is cut by S}\Big]
	&= \Ex_{i,j}\Ex_{u}\Ex_{v_i,v_j \sim N_{\cG}(u)}\Pr_{(x',y')}\Big[f(v_i,i,\pi_{v_i,u}(x')) \neq f(v_j,j,\pi_{v_j,u}(y'))\Big] 
	\non \\
	&= \Ex_{i,j}\Ex_{u}\Ex_{v_i,v_j \sim N_{\cG}(u)}\Ex_{(x',y')}
	\Big[{\sf Cut}\Big(f_{v_i,i}(\pi_{v_i,u}(x')),f_{v_j,j}(\pi_{v_j,u}(y'))\Big)\Big] 		\non \\
	&= \Ex_{i,j}\Ex_{u,v_i,v_j}\Ex_{(x,y)}\Ex_{(x',y')|(x,y)}\Big[{\sf Cut}\Big(f_{v_i,i}(\pi_{v_i,u}(x')),f_{v_j,j}(\pi_{v_j,u}(y'))\Big)\Big]		\non  \\
	&= \Ex_{i,j}\Ex_{u,v_i,v_j}\Ex_{(x,y)}\Big[{\sf Cut}\Big(T_{1 - \eta}f_{v_i,i}(\pi_{v_i,u}(x)),T_{1 - \eta}f_{v_j,j}(\pi_{v_j,u}(y))\Big)\Big]  		\non \\
	&=  \Ex_{i,j}\Ex_{u \sim U}\Ex_{(x,y)}\Big[{\sf Cut}\Big(g_{u,i}(x),g_{u,j}(y)\Big)\Big],			\label{eqn:sound-1}
\end{align}
where for every $i \in [3]$ and $u \in U$, we define $g_{u,i}:\{0,1\}^R \to [0,1]$ to be the averaged function
\begin{equation}
	g_{u,i}(x) := \Ex_{v \sim N_{\cG}(u)}\Big[T_{1 - \eta}f_{v,i}\Big(\pi_{v,u}(x)\Big)\Big].
\end{equation}

\subsubsection{Main Technical Step}

Let us define $U_{\rm good} \subset U$ as
\[
U_{\rm good} := \left\{u \in U : \max_{i \in [3]}\max_{\ell \in [R]} \Inf{\ell}{g_{u,i}} \leq \tau \right\}.
\]

The following lemma shows that if $\cG$ is a NO instance, then most of the vertices in $U$ are good.
\begin{lemma}				\label{lem:inf-dec}
	Suppose ${\sf val}(\cG) \leq \delta$. Then $|U_{\rm good}| \geq (1 - \eta')|U|$, where $\eta' = \eta'(\delta,\tau,\eta)$.
\end{lemma}

The above lemma is established using standard influence-decoding arguments, for the sake of completeness we include a proof in Section \ref{sec:dec}. The next lemma bounds the expectation term for each choice of $u \in U_{\rm good}$. 

\begin{lemma}				\label{lem:g-bound}
	For every $u \in U_{\rm good}$, we have
	\[
	\Ex_{i \neq j}\Ex_{x\underset{\rho^*}{\sim} y}\left[{\sf Cut}(g_{u,i}(x),g_{u,j}(y))\right] \leq \frac{\arccos \rho^*}{\pi} + c(\tau),
	\]
	where $c(\tau)$ is an increasing function of $\tau$ satisfying $c(0) = 0$.
\end{lemma}
\begin{proof}
	 Note that since $u \in U_{\rm good}$, using the definition of $U_{\rm good}$ it follows that
	 \[
	 \max_{i \in [3]}\max_{\ell \in [R]}\Inf{\ell}{T_{1 - \eta}g_{u,i}} \leq \tau.
	 \]
	 Hence, instantiating Lemma \ref{lem:main} with $g_i = g_{u,i}$ for every $i \in \{1,2,3\}$ yields the bound
	 \[
	 \Ex_{i \neq j}\Ex_{x\underset{\rho^*}{\sim}y}\left[{\sf Cut}(g_{u,i}(x),g_{u,j}(y))\right] \leq \frac{\arccos \rho^*}{\pi} + c(\tau).
	 \]	
\end{proof}
Using the above, we can now continue with the soundness analysis. Continuing from \eqref{eqn:sound-1}, we get that
\begin{align}
	\Ex_{u \sim U}\Ex_{i \neq j}\Ex_{x,y}\Big[{\sf Cut}\left(g_{u,i}(x),g_{u,j}(y)\right)\Big]		
	&\overset{1}{\leq} \Ex_{u \sim U_{\rm good}} \Ex_{i \neq j}\Big[{\sf Cut}(g_{u,i}(x),g_{u,j}(y))\Big] + \eta' 	\non\\
	&\overset{2}{\leq} \frac{\arccos \rho^*}{\pi} + c(\tau) + \eta'			\non \\		
	&\overset{3}{\leq} \frac{\arccos \rho^*}{\pi} + 2\eta			
	\label{eqn:3col-sound}
\end{align}
where in step $1$ we use the facts that $|U_{\rm good}| \geq (1 - \eta')|U|$ using Lemma \ref{lem:inf-dec}, and in step $2$ we use the bound from Lemma \ref{lem:main}, and step $3$ follows from our choice of parameters. This concludes the proof of Lemma \ref{lem:sound-3col}.

\section{Tripartite Noise Stability Bound}

In this section, we prove the tripartite variant of Majority-Is-Stablest, as stated in the lemma below.

\begin{lemma}			\label{lem:main}
	There exists $\tau_0 \in [0,1]$ such that the following holds for any $\tau \in (0,\tau_0]$ and $\rho \in (-1,0]$. Let $g_1,g_2,g_3:\{0,1\}^R \to [0,1]$ be functions satisfying
	\begin{equation}			\label{eqn:inf-bound}
	\max_{i \in [3]}\max_{\ell \in [R]} \Inf{\ell}{g_{i}} \leq \tau.
	\end{equation}
	Then,
	\[
	\Ex_{i \neq j} \Ex_{x \underset{\rho}{\sim} y}\left[{\sf Cut}\Big(g_{i}(x),g_j(y)\Big)\right]
	\leq 
	\begin{cases}
	\frac{\arccos \rho}{\pi} + c(\tau)	& \text{if} -1 < \rho < -1/2, \\
	\frac23	& \text{if} -1/2 \leq \rho \leq 0,	
	\end{cases}
	\]
	where $c(\tau)$ is an increasing continuous function of $\tau$ satisfying $c(0) = 0$. 
	
	Moreover, the following holds:
	\begin{itemize}
		\item When $\rho \in (-1,-1/2)$, the maximum is attained by the assignment $g_i := {\sf MAJ}$ for every $i \in \{1,2,3\}$.
		\item When $\rho \in [-1/2,0]$, the maximum is attained by assignments of the form $(g_1,g_2,g_3)$ where $g_i \equiv 0$,$g_j \equiv 1$, and $g_k:\{0,1\}^R \to \{0,1\}$ is an arbitrary function, for any ordering of $i,j,k \in \{1,2,3\}$
	\end{itemize}
\end{lemma}

The above lemma gives a complete characterization (up to $o(1)$-error terms) of the maximum-cut value that can be obtained using pseudorandom cuts for the tripartite noisy cube gadget, for all $\rho \in [-1,0]$. In particular, it states that for the case when $\rho < -1/2$, then the Majority function, {\sf MAJ}, attains the largest cut value $\frac{\arccos \rho^*}{\pi}$, which is the same as the standard noisy-cube gadget (Theorem \ref{thm:stab}). However, when $\rho \geq -1/2$, then the cut is no longer maximized by the ${\sf Maj}$ function. Instead, for this range of $\rho$, the underlying tripartite structure of the gadget takes over, and the cut-value of the gadget is maximized by degenerate cuts where two of the three copies of the sets are placed on opposite sides i.e., $g_i \equiv 1$ and $g_j \equiv 0$, and the third set, represented by $g_k$, can be chosen arbitrarily. 

\subsection{Preliminaries}

Let us begin by setting up some useful notation. Recall that for any $x,y \in [0,1]$, we define
\[
\Lambda_\rho(x,y) = \Pr_{g_1 \underset{\rho}{\sim} g_2}\Big[g_1 \leq \Phi^{-1}(x), g_2 \leq \Phi^{-1}(y)\Big].
\]
Furthermore, for a fixed $\rho \in [-1,1]$, we define the {\em $\rho$-correlated Gaussian Cut function}, denoted by ${\sf Cut}_\rho:[0,1]^2 \to [0,1]$, in the following way:
\[
{\sf Cut}_\rho(\mu_1,\mu_2) := \Pr_{g_1 \underset{\rho}{\sim} g_2} \Big[\one_{H_{\mu_1}}(g_1) \neq \one_{H_{\mu_2}}(g_2)\Big],
\]
where $H_\mu := \{x \in \R : x \leq \Phi^{-1}(\mu) \}$ is the halfspace of Gaussian measure $\mu$. In other words, ${\sf Cut}_\rho(\mu_1,\mu_2)$ is the probability of a pair of $\rho$-correlated Gaussian variables getting cut by halfspaces of measures $\mu_1$ and $\mu_2$ respectively. We will find it convenient to work with the following formula for ${\sf Cut}_\rho(\cdot,\cdot)$ stated in the lemma below.
\begin{lemma}			\label{lem:cut-def}
	For any $\mu_1,\mu_2 \in [0,1]$ we have 
	\[
	{\sf Cut}_{\rho}(\mu_1,\mu_2) = 1 - \Lambda_\rho(\mu_1,\mu_2) - \Lambda_{\rho}(1 - \mu_1,1 - \mu_2).
	\]
\end{lemma}
\begin{proof}
	Note that by definition, for any $a,b \in \{0,1\}$, we have ${\sf Cut}(a,b) = a + b - 2ab$. Hence, we can write
	\begin{align*}
		&\Ex_{g_1 \underset{\rho}{\sim} g_2}\Big[{\sf Cut}\left(H_{\mu_1}(g_1),H_{\mu_2}(g_2)\right)\Big] \\
		&= \Ex_{g_1 \underset{\rho}{\sim} g_2}\Big[H_{\mu_1}(g_1) + H_{\mu_2}(g_2) - 2H_{\mu_1}(g_1)H_{\mu_2}(g_2)\Big] \\
		&= \Ex_{g_1 \underset{\rho}{\sim} g_2}\Big[1 - H_{\mu_1}(g_1)H_{\mu_2}(g_2) - (1 - H_{\mu_1}(g_1))(1 - H_{\mu_2}(g_2))\Big] \\
		&= 1 - \Ex_{g_1 \underset{\rho}{\sim} g_2}\Big[H_{\mu_1}(g_1)H_{\mu_2}(g_2)\Big] - \Ex_{g_1 \underset{\rho}{\sim} g_2}\Big[(1 - H_{\mu_1})(g_1)(1 - H_{\mu_2})(g_2)\Big] \\
		&= 1 - \Lambda_\rho(\mu_1,\mu_2) - \Lambda_{\rho}(1 - \mu_1,1 - \mu_2).
	\end{align*}
\end{proof}

We now state a lemma which is the main technical tool in proving Lemma \ref{lem:main}.

\begin{lemma}				\label{lem:str-main}
	The following holds:
	\[
	\max_{\mu_1,\mu_2,\mu_3 \in [0,1]} \Ex_{i \neq j}\Big[{\sf Cut}_\rho(\mu_i,\mu_j)\Big] = 
	\begin{cases}
		\frac{\dsp \arccos \rho}{\dsp \pi}	&\mbox{ if } -1 < \rho < -1/2,		\\
		\frac{\dsp 2}{\dsp 3}	&\mbox{ if } -1/2 \leq \rho \leq 0.
	\end{cases}
	\]
	Furthermore, the following holds:
	\begin{itemize}
		\item When $\rho \in (-1,-1/2)$, the maxima is uniquely attained at $\mu_1 = \mu_2 = \mu_3 = 1/2$.
		\item When $\rho \in [-1/2,0]$, the maxima is attained by assignments of the form $\mu_i = 0,\mu_j = 1$ and $\mu_k \in [0,1]$ where $\{i,j,k\} = \{1,2,3\}$. 
	\end{itemize}
\end{lemma}

In other words, the above lemma characterizes the optimal choice of hyperplanes which achieves the maximum cut in the ``Gaussian analogue'' of the tripartite gadget. In particular, it says that if $\rho < -1/2$, then the optimal cut corresponds to choosing balanced thresholds in each of the three partitions. And as before, if $\rho \geq -1/2$, then the optimal choice corresponds to the degenerate cut where two out of the three thresholds are put in opposite sides, and the threshold for the third set is chosen arbitrarily. We defer the proof of the above lemma to Section \ref{sec:str-main} for now, and proceed with proving Lemma \ref{lem:main}.

\subsection{Proof of Lemma \ref{lem:main}}

Now we are ready to prove Lemma \ref{lem:main}. For every $g_i$, let us denote $\bg_i := 1 - g_i$ and let $\mu_i = \Ex_{x \sim \{0,1\}^R}\big[g_i(x)\big]$. Then using the definition of ${\sf Cut}(\cdot,\cdot)$ we get that:
\begin{align*}
	\Ex_{i \neq j}\Ex_{x \underset{\rho}{\sim} y}\Big[{\sf Cut}\big(g_i(x),g_j(y)\big)\Big]
	&= \Ex_{i \neq j}\Ex_{x \underset{\rho}{\sim} y}\Big[1 - g_i(x)g_j(y) - \bg_i(x)\bg_j(y)\Big]	\\
	&= 1 - \Ex_{i \neq j}\Ex_{x \underset{\rho}{\sim} y}\Big[g_i(x)g_j(y)\Big] - \Ex_{i \neq j}\Ex_{x \underset{\rho}{\sim} y}\Big[\bg_i(x)\bg_j(y)\Big].
\end{align*}
Now fix a pair $(i,j)$. Then, due to the setting of this lemma we have
\[
\max_{j \in [R]}\Inf{j}{g_a} \leq \tau
\]
for every $a \in \{1,2,3\}$. Therefore, using Theorem \ref{thm:stab} we have
\[
\Ex_{x \underset{\rho}{\sim} y} \Big[g_i(x)g_j(y)\Big]
\geq \Lambda_\rho\left(\Ex[g_i],\Ex[g_j]\right) - c(\tau) = \Lambda_\rho(\mu_i,\mu_j) - c(\tau).
\]
By an identical argument, we also get that
\[
\Ex_{x \underset{\rho}{\sim} y} \Big[\bg_i(x)\bg_j(y)\Big] 
\geq \Lambda_\rho\Big(\Ex[\bg_i],\Ex[\bg_j]\Big) - c(\tau)
= \Lambda_\rho\big(1 - \mu_i,1 - \mu_j\big) - c(\tau).
\]
Putting together the bounds, for every fixed choice of $i,j$ we get that
\begin{align*}
	\Ex_{x \underset{\rho}{\sim} y}\Big[{\sf Cut}(g_i(x),g_j(y))\Big] 
	&\leq 1 - \Lambda_\rho\big(\mu_i,\mu_j\big) - \Lambda_\rho\big(1 - \mu_i,1 - \mu_j\big) + 2 c(\tau)  	\\
	&= {\sf Cut}_\rho(\mu_i,\mu_j) + 2c(\tau).
\end{align*}

Hence, applying the above bound for every fixed choice of $i,j$ we get that
\begin{equation}			\label{eqn:tri-1}
	\Ex_{i,j}\Ex_{x \underset{\rho}{\sim} y}\Big[{\sf Cut}(g_i(x),g_j(y))\Big]
	\leq \Ex_{i,j}\Big[{\sf Cut}_\rho\big(\mu_i,\mu_j\big)\Big] + 2 c(\tau).
\end{equation}
Finally, using Lemma \ref{lem:str-main}, we get the bounds:
\begin{equation}			\label{eqn:tri-2}
	\Ex_{i \neq j}\Big[{\sf Cut}_\rho(\mu_i,\mu_j)\Big] \leq
	\begin{cases}
		\frac{\dsp \arccos \rho}{\dsp \pi} & \mbox{ if } -1 < \rho < -1/2,			\\
		\frac{\dsp 2}{\dsp 3} & \mbox { if } -1/2 \leq \rho \leq 0.
	\end{cases}
\end{equation} 

Combining the bounds from \eqref{eqn:tri-1} and \eqref{eqn:tri-2} yields the bound:
\[
\Ex_{i \neq j} \Ex_{x \underset{\rho}{\sim} y}\left[{\sf Cut}\Big(g_{i}(x),g_j(y)\Big)\right]
\leq 
\begin{cases}
	\frac{\arccos \rho}{\pi} + c(\tau)	& \text{if} -1 < \rho < -1/2, \\
	\frac23	& \text{if} -1/2 \leq \rho \leq 0,	
\end{cases}
\]
which establishes the first part of the lemma. For the second part of the lemma, we observe the following:
\begin{itemize}
	\item Using Sheppard's Formula (Lemma \ref{lem:shep}), the assignment $g_1 = g_2 = g_3 = {\sf MAJ}$ achieves the cut
	\[
	\Ex_{i\neq j}\Ex_{x \underset{\rho}{\sim} y}\left[{\sf Cut}(g_i(x),g_j(y))\right]
	= \Ex_{x \underset{\rho}{\sim} y}\left[{\sf Cut}(g_i(x),g_j(y))\right]
	= \frac{\arccos \rho}{\pi} \pm o_R(1).
	\]
	\item Any choice of functions $(g_i,g_j,g_k)$ where $g_i \equiv 1$, $g_j \equiv 0$, and $g_k:\{0,1\}^R \to \{0,1\}$ is an arbitrary function satisfying \eqref{eqn:inf-bound} achieves a cut of value $2/3$.
\end{itemize}
Furthermore, in both cases (choosing $R$ large enough), the above choice of assignments $(g_1,g_2,g_3)$ satisfy \eqref{eqn:inf-bound} \footnote{In particular, it is known that all the influences of the Majority function are $\Theta(1/\sqrt{R})$, for e.g., see Section 5.2 of \cite{OD14}. Furthermore, the influences of constant functions are trivially $0$.}. Hence, the claim follows.

\section{Proof of Lemma \ref{lem:str-main}}            \label{sec:str-main}

In this section, we prove Lemma \ref{lem:str-main}. Let us introduce some additional useful notations which will be used in the proof. Throughout this section, let us denote
\[
F_{\rm cut}(\mu_1,\mu_2,\mu_3) := \sum_{\substack{i,j \in [3] \\ i < j}} {\sf Cut}_\rho(\mu_i,\mu_j).
\]
Furthermore, we will also define
\[
\ccut(x,y) := 1 - {\sf Cut}_\rho(x,y),
\]
i.e., $\ccut(x,y) = \Lambda_\rho(x,y) + \Lambda_\rho(1 - x,1 - y)$. Finally, for any $\mu_1,\mu_2,\mu_3 \in [0,1]$, let us define
\[
F(\mu_1,\mu_2,\mu_3) := \sum_{\substack{i,j \in [3] \\ i < j}} \ccut(\mu_i,\mu_j).
\]
Note that we have the set of maximizers of $F_{\rm cut}$ are precisely the set of minimizers of $F$, i.e.,
\[
\argmax_{\mu_1,\mu_2,\mu_3 \in [0,1]} F_{\rm cut}(\mu_1,\mu_2,\mu_3) = \argmin_{\mu_1,\mu_2,\mu_3 \in [0,1]} F(\mu_1,\mu_2,\mu_3),
\]
and hence it suffices to characterize the minima(s) of $F$ instead\footnote{We work with $F$ instead of $F_{\rm cut}$ for the most part, since $F$ doesn't carry the extra constant terms.}. To that end, we will show the following:
\begin{itemize}
	\item For $\rho \in (-1,0]$ with $\rho \neq -1/2$, the balanced-cuts assignment $(1/2,1/2,1/2)$ is the only stationary point of $F$ in $(0,1)^3$ (Lemma \ref{lem:dv-int}).
	\item On the boundary of $[0,1]^3$, the extreme values of $F$ are attained by triples of the form $(0,1,\mu)$ where $\mu \in [0,1]$, and all permutations thereof (Lemma \ref{lem:int-val}).
\end{itemize}
We now proceed to establish the above in the subsequent sections.

\subsection{Interior Stationary Points}

The main result of this section is the following lemma which characterizes the interior stationary points of $F$.
\begin{lemma}			\label{lem:dv-int}
	Suppose $\rho \in (-1,0]$ such that $\rho \neq -1/2$. Then the only stationary point of $F(\mu_1,\mu_2,\mu_3)$ in $(0,1)^3$ is $\mu_1 = \mu_2 = \mu_3 = 1/2$.
\end{lemma}

We prove the above lemma in the next couple of subsections.

\subsection{Setting Up}

For convenience, throughout this section we shall denote $t(x) := \Phi^{-1}(x)$, where recall that $\Phi^{-1}(\cdot)$ is the inverse Gaussian CDF function. The following expression for the partial derivatives of $\Lambda_\rho$ is well-known.

\begin{lemma}			\label{lem:dv}
	For any $x,y \in [0,1]$, we have
	\[
	\frac{\partial \Lambda_\rho(x,y)}{\partial x}  = \Phi\left(\frac{t(y) - \rho t(x)}{\sqrt{1 - \rho^2}}\right).
	\]
\end{lemma}

We defer the proof of the lemma to Section \ref{sec:proof-dv}. Using the above, we get the following expression for the derivative of $\ccut(\cdot,\cdot)$ which is a key tool in the proof.

\begin{lemma}			\label{lem:dv-1}
	For any $x,y \in [0,1]$, we have
	\[
	\frac{\partial \ccut(x,y)}{\partial x} 
	= \Phi\left(\frac{t(y) - \rho t(x)}{\sqrt{1 - \rho^2}}\right) - \Phi\left(-\frac{t(y) - \rho t(x)}{\sqrt{1 - \rho^2}}\right).
	\]
\end{lemma}
\begin{proof}
	Using Lemma \ref{lem:dv} and the definition of $\ccut$ we have:
	\begin{align*}
		\frac{\partial \ccut(x,y)}{\partial x} 
		&= \frac{\partial \Lambda_\rho(x,y)}{\partial x}  + \frac{\partial \Lambda_\rho(1 - x,1 - y)}{\partial x} \\
		&= \Phi\left(\frac{t(y) - \rho t(x)}{\sqrt{1 - \rho^2}}\right) + \frac{d (1 - x)}{dx} \cdot \Phi\left(\frac{t(1 - y) - \rho t(1 - x)}{\sqrt{1 - \rho^2}}\right) \\
		&= \Phi\left(\frac{t(y) - \rho t(x)}{\sqrt{1 - \rho^2}}\right) - \Phi\left(\frac{t(1 - y) - \rho t(1 - x)}{\sqrt{1 - \rho^2}}\right) \\
		&= \Phi\left(\frac{t(y) - \rho t(x)}{\sqrt{1 - \rho^2}}\right) - \Phi\left(- \frac{t(y) - \rho t(x)}{\sqrt{1 - \rho^2}}\right),
	\end{align*}
	where the last step is due to $\Phi^{-1}(1 - z) = - \Phi^{-1}(z)$ for all $z \in [0,1]$.
\end{proof}

For ease of notation, let us also define
\[
K(x,y) := \frac{t(y) - \rho t(x)}{\sqrt{1 - \rho^2}},
\]
and 
\[
G(z) := \Phi(z) - \Phi(-z).
\]
We use the above notation to express the partial derivatives of the function $F$.

\begin{lemma}			\label{lem:dv-2}
	For any $\mu_1,\mu_2,\mu_3 \in [0,1]$ we have
	\[
	\frac{\partial  F(\mu_1,\mu_2,\mu_3)}{\partial \mu_1} = G(K(\mu_1,\mu_2)) + G(K(\mu_1,\mu_3)).
	\]
	Similarly, $\frac{\partial F}{\partial \mu_2},\frac{\partial F}{\partial \mu_3}$ admit analogous formulas.
\end{lemma}
\begin{proof}
	Using the definition of $F$, we get 
	\begin{equation}		\label{eqn:dv-1a}
		\frac{\partial F(\mu_1,\mu_2,\mu_3)}{\partial \mu_1} 
		= \frac{\partial \ccut(\mu_1,\mu_2)}{\partial \mu_1} + \frac{\partial \ccut(\mu_1,\mu_3)}{\partial \mu_1} 	\\
	\end{equation}	
	Now, applying Lemma \ref{lem:dv-1}, we can evaluate the first term as
	\begin{align}
		\frac{\partial \ccut(\mu_1,\mu_2)}{\partial \mu_1}
		&= \Phi\left(\frac{t(\mu_2) - \rho t(\mu_1)}{\sqrt{1 - \rho^2}}\right) - \Phi\left(-\frac{t(\mu_2) - \rho t(\mu_1)}{\sqrt{1 - \rho^2}}\right)
		\non    \\
		&= \Phi\left(K(\mu_1,\mu_2)\right) - \Phi\left(-K(\mu_1,\mu_2)\right)		\non \\
		&= G(K(\mu_1,\mu_2)),			\label{eqn:dv-1b}
	\end{align}
	where the second step uses the definition of $K(\cdot,\cdot)$ and the third step uses the definition of $G(\cdot)$. By a similar argument, we also have
	\begin{equation}			\label{eqn:dv-1c}
	\frac{\partial \ccut(\mu_1,\mu_3)}{\partial \mu_1} = G(K(\mu_1,\mu_3)).
	\end{equation}
	Plugging in the expressions from \eqref{eqn:dv-1b} and \eqref{eqn:dv-1c} into \eqref{eqn:dv-1a} we immediately get that
	\[
	\frac{\partial F(\mu_1,\mu_2,\mu_3)}{\partial \mu_1} = G(K(\mu_1,\mu_2)) + G(K(\mu_1,\mu_3)),
	\]
	as claimed by the lemma. The partial derivatives for $\mu_2,\mu_3$ can be derived using similar arguments.	
\end{proof}

{\bf Some Useful Properties of $G(\cdot)$}. We also state and prove some useful properties of the function $G(\cdot)$.

\begin{observation}			\label{obs:dv-1}
	The function $G(\cdot)$ is odd and strictly increasing on $\R$.
\end{observation}
\begin{proof}
	Since $G(z) = \Phi(z) - \Phi(-z)$, $G$ is odd by definition. Furthermore, since $\Phi(z)$ is strictly increasing in $z$, it follows that $G(z) = \Phi(z) - \Phi(-z)$ is strictly increasing in $z$ as well.
\end{proof}

\begin{observation}			\label{obs:dv-2}
	A pair $z_1,z_2 \in \R$ satisfies $G(z_1) + G(z_2) = 0$ if and only if $z_1 + z_2 = 0$.
\end{observation}
\begin{proof}
	Since $G(\cdot)$ is strictly increasing (Observation \ref{obs:dv-1}), the inverse function $G^{-1}$ is well-defined. Hence $G(z_1) + G(z_2) = 0$ if and only if 
	\[
	z_1 = G^{-1}(-G(z_2)) = G^{-1}(G(- z_2)) = -z_2,
	\]	
	where the middle inequality uses the fact that $G$ is odd (Observation \ref{obs:dv-1}).
\end{proof}
 
\subsection{Proof of Lemma \ref{lem:dv-int}}

Now we are ready to prove Lemma \ref{lem:dv-int}.

\begin{proof}[Proof of Lemma \ref{lem:dv-int}]
	Fix $i \in \{1,2,3\}$. Then using Lemma \ref{lem:dv-2} we get that
	\[
	\frac{\partial F(\mu_1,\mu_2,\mu_3)}{\partial \mu_i}
	= G(K(\mu_i,\mu_j)) + G(K(\mu_i,\mu_k))
	\]
	where $j,k\neq i$. Now if $(\mu_1,\mu_2,\mu_3)$ is a stationary point of $F$, it follows that 
	\[
	0 = \frac{\partial F(\mu_1,\mu_2,\mu_3)}{\partial \mu_i} = G(K(\mu_i,\mu_j)) + G(K(\mu_i,\mu_k)),
	\]
	for all $i \in \{1,2,3\}$. But then using Observation \ref{obs:dv-2}, this implies that
	\[
	K(\mu_i,\mu_j) + K(\mu_i,\mu_k) = 0,
	\]
	for every $i \neq j \neq k$. Now let us take the constraint $K(\mu_1,\mu_2) + K(\mu_1,\mu_3) = 0$. Using the definition of $K(\cdot,\cdot)$, this implies 
	\[
	\frac{t(\mu_2) - \rho t(\mu_1)}{\sqrt{1 - \rho^2}} + \frac{t(\mu_3) - \rho t(\mu_1)}{\sqrt{1 - \rho^2}} = 0,
	\]
	which for $\rho \neq \pm 1$ implies that
	\[
	t(\mu_2) + t(\mu_3) - 2\rho \cdot t(\mu_1) = 0.
	\]
	Similarly, using the conditions $\frac{\partial F}{\partial \mu_2} = 0$ and $\frac{\partial F}{\partial \mu_3} = 0$, we also get the constraints: 
	\[
	t(\mu_1) + t(\mu_3) - 2\rho \cdot t(\mu_2) = 0
	\qquad\qquad
	\textnormal{and}
	\qquad\qquad
	t(\mu_1) + t(\mu_2) - 2\rho \cdot t(\mu_3) = 0.
	\]
	Hence if $(\mu_1,\mu_2,\mu_3)$ is a stationary point of $F$, then $t_i := t(\mu_i)$ for $i = 1,2,3$ must satisfy the constraints:
	\begin{equation}			\label{eqn:eq-term}
	\begin{bmatrix}
		-2 \rho & 1 & 1 \\
		1 & -2\rho & 1 \\
		1 & 1 & - 2 \rho 
	\end{bmatrix}
	\begin{bmatrix}
		t_1  \\  t_2  \\ t_3
	\end{bmatrix}
	= 
	\begin{bmatrix}
		0  \\  0  \\ 0
	\end{bmatrix}
	.
	\end{equation}
	Note that when $\rho \neq -1/2$, the coefficient matrix in the LHS is full-rank. Hence $t_1 = t_2 = t_3 = 0$ is the only solution to the above system, which implies that we have $\mu_i = \Phi(t_i) = 1/2$ for every $i \in \{1,2,3\}$.
\end{proof}

\subsection{Maximizing $F_{\rm cut}(\cdot)$ at the boundary of $[0,1]^3$}

Let $\cB \subset [0,1]^3$ denote the boundary of $[0,1]^3$ i.e.,
\[
\cB = \left\{(\mu_1,\mu_2,\mu_3) \in [0,1]^3~\Big|~\exists~i \in [3]~\text{\it s.t. } \mu_i \in \{0,1\}\right\}.
\]
The following lemma characterizes the maximizers of $F_{\rm cut}$ on the boundary of $[0,1]^3$.

\begin{lemma}			\label{lem:int-val}
	The set of maximizers of the function $F_{\rm cut}$ within $\cB$ is exactly the set
	\[
	\cS^* := \left\{(\mu_1,\mu_2,\mu_3) \in [0,1]^3~\Big|~ \exists~i,j \in [3] \text{ s.t. } \mu_i = 0,\mu_j = 1\right\},
	\]
	and the corresponding maximum value is $2$.  
\end{lemma}
\begin{proof}
	Let $\bmu = (\mu_1,\mu_2,\mu_3)$ be a maximizer of $F_{\rm cut}$ within $\cB$. Since $\bmu \in \cB$, without loss of generality we may assume\footnote{Otherwise, we can work with $1 - \bmu = (1 - \mu_1,1 - \mu_2,1 - \mu_3)$, since $F_{\rm cut}(1 - \bmu) = F_{\rm cut}(\bmu)$.} that $\mu_1 = 0$. We shall need the following useful claim.
	
	\begin{claim}				\label{cl:dv}
		Suppose $\bmu$ is a maximizer of $F_{\rm cut}$ within $\cB$ such that $\mu_1 = 0$. Then we must have $\max\{\mu_2,\mu_3\} = 1$.
	\end{claim}
	\begin{proof}
		Suppose for contradiction, let us assume that $\max\{\mu_2,\mu_3\} < 1$. We consider two cases.
		
		{\bf Case (i)}. Suppose $\mu_2 = \mu_3 = 0$. Then note that $F_{\rm cut}(\bmu) = 0$; however note that that the maximum value of $F_{\rm cut}$ in $\cB$ is strictly positive (for e.g., $F_{\rm cut}(0,1,0) = 2$), and hence $\bmu = (0,0,0)$ cannot be a maximizer of $F_{\rm cut}$.
		
		{\bf Case (ii)}. Suppose $\max\{\mu_2,\mu_3\} > 0$. Then without loss of generality, we may assume that $\mu_2 \in (0,1)$. Since $\bmu$ is a maximizer of $F_{\rm cut}$ (and hence, a minimizer of $F$) within $\cB$, it follows that $\bmu$ satisfies:
		\[
		\frac{\partial F(\mu_1,\mu_2,\mu_3)}{\partial \mu_2} = 0.
		\] 
		On the other hand, note that using Lemma \ref{lem:dv-2}, denoting $L = \sqrt{1 - \rho^2}$ we have
		\begin{align*}
			&\frac{\partial F(\mu_1,\mu_2,\mu_3)}{\partial \mu_2}  \\
			&= \frac{\partial \ccut(\mu_1,\mu_2)}{\partial \mu_2} + \frac{\partial \ccut(\mu_3,\mu_2)}{\partial \mu_2} \\
			&= \Phi\left(\frac{t(0) - \rho t(\mu_2)}{L}\right) - \Phi\left(-\frac{t(0) - \rho t(\mu_2)}{L}\right)
			+ \Phi\left(\frac{t(\mu_3) - \rho t(\mu_2)}{L}\right) - \Phi\left(-\frac{t(\mu_3) - \rho t(\mu_2)}{L}\right) 	\\
			&= -1 + G\left(\frac{t(\mu_3) - \rho t(\mu_2)}{L}\right). 
		\end{align*}
		Now, if $\frac{\partial F}{\partial \mu_2} = 0$, then the above computation implies that
		\[
		G\left(\frac{t(\mu_3) - \rho t(\mu_2)}{L}\right) = 1.
		\]
		Furthermore, since $G$ is strictly increasing and $\rho \leq 0$, the above can happen only when either $\mu_2 = 0$ or $\mu_3 = 1$. Note that the former cannot happen (since $\mu_2 > 0$ by assumption), which leads to conclusion that $\mu_3 = 1$. This in turn again contradicts our assumption that $\max\{\mu_2,\mu_3\} < 1$. 
		
		In summary, since both cases (i) and (ii) yield a contradiction, our claim follows.
	\end{proof}
	
	Now proceeding with the proof of Lemma \ref{lem:int-val}, given that $\bmu$ is a maximizer of $F_{\rm cut}$ within $\cB$ and $\mu_1 = 0$, the above lemma implies that either $\mu_2 = 1$ or $\mu_3 = 1$. Let us assume $\mu_2 = 1$. Then, note that for any $\mu'_3 \in [0,1]$ we have
	\begin{align*}
		F_{\rm cut}(\mu_1,\mu_2,\mu'_3) 
		&= {\sf Cut}_\rho(0,1) + {\sf Cut}_\rho(0,\mu'_3) + {\sf Cut}_\rho(1,\mu'_3) \\
		&= 1 + \mu'_3 + (1 - \mu'_3) \\
		&= 2,
	\end{align*}
	i.e., $F_{\rm cut}$ attains the maximum value $2$ for all choices of $\mu'_3 \in [0,1]$. Now to conclude the proof, observe that $F_{\rm cut}$ is invariant under the permutation of its three inputs, and hence all permutations of $(0,1,\mu'_3)$ with $\mu'_3 \in [0,1]$ are also maximizers of $F_{\rm cut}$. Since the set $\cS^*$ is precisely the set of all such triples $(\mu_1,\mu_2,\mu_3)$, this concludes the proof.
	
\end{proof}

\subsection{Proof of Lemma \ref{lem:str-main}}

Let $\rho \in (-1,0]$ be such that $\rho \neq -1/2$. Let $(\mu^*_1,\mu^*_2,\mu^*_3)$ be the maximizer of $F_{\rm cut}(\cdot)$ over $[0,1]^3$. We consider two cases. 

{\bf Case (i)}: Suppose $(\mu^*_1,\mu^*_2,\mu^*_3) \in (0,1)^3$ is an interior point. Then $(\mu^*_1,\mu^*_2,\mu^*_3)$ must be a stationary point of $F$ and hence using Lemma \ref{lem:dv-2}, we have that $\mu^*_i = 1/2$ for $i = 1,2,3$. This implies that:
\[
F_{\rm cut}(\mu^*_1,\mu^*_2,\mu^*_3) 
= \sum_{i < j}{\sf Cut}_\rho\left(\frac12,\frac12\right)
= 3\cdot\frac{\arccos \rho}{\pi}.
\]

{\bf Case (ii)}: If $(\mu^*_1,\mu^*_2,\mu^*_3)$ is not an interior point, then using Lemma \ref{lem:int-val} we must have
\[
F_{\rm cut}(\mu^*_1,\mu^*_2,\mu^*_3) = 2,
\]
and it is attained by all assignments of the form $\mu^*_i = a$,  $\mu^*_j = 1- a$ and $\mu^*_k \in [0,1]$ with $a \in \{0,1\}$.

{\bf Finishing the proof.} Putting the above cases together we get that
\[
F_{\rm cut}(\mu^*_1,\mu^*_2,\mu^*_3)
= 3\cdot\max\left\{\frac{\arccos \rho}{\pi},\frac{2}{3}\right\}.
\]
In particular, if $\rho < -1/2$, then the maxima is achieved by the assignment $(1/2,1/2,1/2)$ and its corresponding value is $3(\arccos \rho)/\pi$. On the other hand, if $\rho \in [-1/2,0]$, the maxima is achieved by assignments of the form given by case (ii), and the corresponding value is $2$. This concludes the proof of Lemma \ref{lem:str-main}.

\section{Proof of Theorem \ref{thm:i-set}}

In this section, we prove Theorem \ref{thm:i-set}.

\subsection{Distributions for the Reduction}

In our reduction, we will rely on two distributions which we shall define here. Let $\mu_1$ be the joint distribution over $\{0,1\} \times \{0,1\}$ such that each bit is uniformly distributed, and the pair of bits is $\rho^*$-correlated, where $\rho^* \approx -0.68915$ is the critical correlation from the Max-Cut analysis.

For the second distribution, fixing a parameter $\alpha \in [0,1]$, let us define $\mu_{2,\alpha}$ as the following joint distribution over a pair of variables:

\begin{equation}
	\mu_{2,\alpha}(x,y) = 
	\begin{cases}
		1 - \alpha & \text{ if } x = y = 0 \\
		\alpha/2   & \text{ if } x = 1, y = 0 \\
		\alpha/2   & \text{ if } y = 1, x = 0 \\
		0 		   & \text{ if } x = y = 1.	
	\end{cases}
\end{equation}

We shall drop the indexing by $\alpha$ and denote $\mu_{2,\alpha}$ as $\mu_2$ whenever $\alpha$ is clear from context.  Finally let us define the measure $\mu: = \mu_1 \otimes \mu_2$.

{\bf Some Useful Properties}. We list some useful properties of the distribution $\mu_{2,\alpha}$.

\begin{lemma}				\label{lem:dist}
	Fixing $\alpha \in [0,1]$, let $(x,y) \sim \mu_{2,\alpha}$. Then,
	\begin{itemize}
		\setlength\itemsep{0.8em}
		\item $\Ex_{\mu_{2,\alpha}}[x] = \Ex_{{\mu_{2,\alpha}}}[y] = \frac{\alpha}{2}$.
		\item $\Pr_{(x,y) \sim \mu_{2,\alpha}}\big[x \neq y \big] = \alpha$.
		\item $\rho_{\alpha} := {\rm Corr}_{\mu_{2,\alpha}}(x,y) = \frac{-\alpha}{2 - \alpha}$.
	\end{itemize}
\end{lemma}
\begin{proof}
	The first two items follow immediately from the definition of $\mu_{2,\alpha}$. For the third item, we have
	\[
	{\rm Corr}_{\mu_{2,\alpha}}(x,y) 
	= \frac{\Ex [xy] - \Ex[x]\Ex[y]}{\sigma(x)\sigma(y)} 
	= \frac{0 - \frac{\alpha^2}{4}}{\frac{\alpha}{2}\left(1 - \frac{\alpha}{2}\right)} 
	= \frac{-\frac{\alpha}{2}}{1 - \frac{\alpha}{2}} 
	= \frac{-\alpha}{2 - \alpha}
	\]
\end{proof}

As a corollary, we get the following:

\begin{corollary}				\label{corr:dist}
	The following holds:
	\begin{itemize}
		\item $\rho_{\alpha}$ is a decreasing function of $\alpha$ satisfying $\rho_0 = 0$.
		\item Let $\alpha^* = \frac{2\rho^*}{\rho^* - 1} \approx 0.81597$. Then $\rho_{\alpha^*} = \rho^*$.
	\end{itemize}
	Consequently, for every $\alpha \in [0,\alpha^*]$ we have $\rho_{\alpha} \in [\rho^*,0]$.
\end{corollary}
\begin{proof}
	Firstly, note that using the third item of Lemma \ref{lem:dist}, we immediately get that $\alpha \mapsto \rho_{\alpha}$ is a decreasing function for $\alpha \geq 0$, and $\rho_0 = 0$. Again, substituting $\alpha = \alpha^*$ in the third item of Lemma \ref{lem:dist} we get that
	\[
	\rho_{\alpha^*} = \frac{-\frac{2\rho^*}{\rho^* - 1}}{2 - \frac{2\rho^*}{\rho^* - 1}} = \rho^*.
	\]
	The last statement now follows immediately by combining the first two items.
\end{proof}

\subsection{The Strong Unique Games Conjecture}

For the proof of Theorem \ref{thm:i-set}, we will need the following stronger version of Conjecture \ref{conj:ugc} as the starting point.

\begin{conjecture}[The Strong Unique Games Conjecture~\cite{KR08}]         \label{conj:strong-ugc}
	There exists constant $\epsilon_0,\delta_0 \in (0,1)$ such that the following holds. For every $\epsilon \in (0,\epsilon_0], \delta \in (0,\delta_0]$ there exists $R = R(\epsilon,\delta)$ such that given a left-regular instance of \ug~instance $\cG$, it is \NP-hard to distinguish between the following cases:
	\begin{itemize}
		\item {\bf YES Case}. There exist a labeling $\sigma: U \cup V \to [R]$ and a subset $V' \subseteq V$ of size at least $(1 - \eps)|V|$ such that $\sigma$ satisfies all the constraints incident on $V'$ in $\cG$.
		\item {\bf NO Case}. No labeling can satisfy more than $\delta$-fraction of constraints in $\cG$.
	\end{itemize}
\end{conjecture}

The above conjecture appears to be slightly stronger than the standard version (Conjecture \ref{conj:ugc}) in the following sense: for the YES case, there exists a labeling which satisfies all the constraints incident on a $(1 - \eps)$-sized subset $V'$. However, it was shown by Khot and Regev~(Theorem 2.3~\cite{KR08}) that the above version is implied by the standard variant (Conjecture \ref{conj:ugc}), and hence we will just choose to work with this version. As we will soon see, this stronger YES case guarantee will be crucial for lifting the ``hidden'' independent set in the dictatorship test to the Unique Games hard instance. 

\subsection{Reduction From \ug}

Now we are ready to describe our reduction from the \ug~problem to Max-Cut.

\begin{figure}[ht!]
	\begin{mdframed}
		Let $\cG(U,V,E,[R],\{\pi_e\}_{e \in E})$ be a $(1 - \epsilon,\delta)$-\ug~instance from Conjecture \ref{conj:strong-ugc}, where the parameters $\epsilon,\delta > 0$ are chosen below the figure. Let $\eta,\alpha > 0$ be constants that are fixed below. We define a new Max-Cut instance $\cH = (V_{\cH},E_{\cH},w_{\cH})$, whose vertices and edges are defined as follows. \\[-5pt]
		
		{\bf Vertex Set}. For every $v \in V$, we introduce a set of vertices
		\[
		V_{\cH,v}:= \left\{(v,x,y) : (x,y) \in \{0,1\}^R\times\{0,1\}^R\right\}.
		\]
		Then we let $V_{\cH} := \cup_{v \in V} V_{\cH,v}$ be the union of all these sets. \\[-5pt]
		
		{\bf Edge Weights}. The weights $w_{\cH}$ on the edges are defined using the following process:
		\vspace{5pt}
		\begin{enumerate}
			\setlength\itemsep{0.5em}
			\item Sample $u \sim U$, and neighbors $v_1,v_2 \sim N_\cG(u)$.
			\item Sample $(x_1,x_2) \sim \mu^{\otimes R}_1$, and $(y_1,y_2) \sim \mu^{\otimes R}_{2,\alpha}$. 
			\item Output\footnote{Here, for a string $(x,y) \in (\{0,1\}\times \{0,1\})^R$, and a permutation $\pi:[R] \to [R]$, the string $\pi(x,y)$ is given by $((x_{\pi(1),y_{\pi(1)}}),(x_{\pi(2)},y_{\pi(2)}),\ldots,(x_{\pi(R)},y_{\pi(R)}))$.} the edge $\left\{\big(v_1,\pi_{v_1,u}(x_1,y_1)\big),\big(v_2,\pi_{v_2,u}(x_2,y_2)\big)\right\}$.
		\end{enumerate}
	\end{mdframed}
	\caption{Reduction for Theorem \ref{thm:i-set}}
	\label{fig:is-red1}
\end{figure}

{\bf Parameters of the Reduction}. We assign the following values to the parameters.
\begin{itemize}
	\item Let $\alpha \in \left[0,\alpha^*\right]$ where recall that $\alpha^*  =\frac{2\rho^*}{\rho^* - 1}$ from Corollary \ref{corr:dist}.
	\item Let $\eta$ be as in Theorem \ref{thm:i-set}.
	\item Let $\tau$ be such that $c(\tau),c_1(\tau) \leq \eta/10$, where $c(\tau),c_1(\tau)$ are as in Theorem \ref{thm:stab} and Lemma \ref{lem:noise}. 
	\item We set $\epsilon = \eta/4$ and $\delta = \tau^2\eta^2/8$.
\end{itemize}

We now state the completeness and soundness guarantees of the above reduction.

\begin{lemma}[Completeness]		\label{lem:comp-iset}
	Suppose ${\sf val}(\cG) \geq 1 - \epsilon$. Then,
	\begin{itemize}
		\item There exists a set $S \subseteq V_{\cH}$ such that $w_\cH(S,S^c) \geq \frac{1 - \rho^*}{2} - 2\eta$.
		\item There exists a set $I \subseteq V_{\cH}$ such that $I$ is an independent set in $\cH$ and $w_{\cH}(I,I^c) \geq \alpha - 2\eta$.
	\end{itemize}
\end{lemma}

\begin{lemma}[Soundness]		\label{lem:sound-iset}
	Suppose ${\sf val}(\cG) \leq \delta$. Then, for every subset $S \subseteq V_{\cH}$ we have 
	\[
		w_{\cH}(S,S^c) \leq \frac{\arccos \rho^*}{\pi} + 2\eta.
	\]
\end{lemma}

As earlier, before we prove the above lemmas, we use them to prove Theorem \ref{thm:i-set}.

\begin{proof}[Proof of Theorem \ref{thm:i-set}]
	Fix constants $\eta > 0$ and $\alpha \in [0,\alpha^*]$. Let $\cG$ be a $(1 - \epsilon,\delta)$-\ug~instance as in Conjecture \ref{conj:strong-ugc} where $\epsilon,\delta$ are chosen as above as functions of $\eta$. Then using the reduction in Figure \ref{fig:is-red1}, we can construct a Max-Cut instance $\cH = (V_\cH,E_{\cH},w_{\cH})$ such that the following holds:
	\begin{itemize}
		\item If $\cG$ is a YES instance, then ${\sf opt}(\cH) \geq c^* - 2\eta$. Moreover, $\cH$ contains an independent set $I \subseteq V_{\cH}$ such that $w_{\cH}(I,I^c) \geq \alpha - 2\eta$.
		\item If $\cG$ is a NO instance, then ${\sf opt}(\cH) \leq s^* + 2\eta$.
	\end{itemize}
	
	Furthermore, for a fixed constant $\eta > 0$, the reduction in Figure \ref{fig:is-red1} runs in time $O_{\eta}({\rm poly}(|U|,|V|))$. Since assuming UGC, the $(1 - \epsilon,\delta)$-\ug~problem is \NP-hard, it follows that it is $\NP$-hard to distinguish between the cases:
	
	{\bf YES Case}. There exists an independent set $I \subseteq V_{\cH}$ such that $w_{\cH}(I,I^c) \geq \alpha - 2\eta$, and ${\sf opt}(\cH) \geq c^* - 2\eta$.
	
	{\bf NO Case}. ${\sf opt}(\cH) \leq s^* + 2\eta$.

	This concludes the proof of Theorem \ref{thm:i-set}.
\end{proof}

\subsection{Proof of Lemma \ref{lem:comp-iset}}

We state the completeness guarantee in the lemma below.

\begin{proof}
	Let $\sigma: U \cup V \to [R]$ be the labeling and $V' \subseteq V$ be the subset of size $|V'| \geq (1 -\eps)|V|$ such that $\sigma$ satisfies all the constraints incident on $V'$ in $\cG$. For the first item, consider the dictator cut defined by the $x$-variables i.e.,
	\[
	S := \Big\{(v,x,y) \in V_{\cH}~\Big|~v \in V' \text{ and } x(\sigma(v)) = 1 \Big\},
	\]
	where recall that $V_{\cH}$ is the vertex set of $\cH$.  Now, to prove the first item we observe that the weight of the edges cut by the partition $(S,S^c)$ is 
	\begin{equation}			\label{eqn:comp-1}
	w_{\cH}(S,S^c) 
	=\Ex_{u}\Pr_{(v_i,x_i,y_i)_{i = 1,2}}\Big[\one_S(v_1,\pi_{v_1,u}(x_1,y_1)) \neq \one_S(v_2,\pi_{v_2,u}(x_2,y_2))\Big] 
	\end{equation}		
	Let $\cE$ denote the event that the vertices $v_1,v_2$ sampled above both belong to $V'$, and note that using the left-regularity of $\cG$, we see that $\Pr[\cE] \geq 1 - 2\eps$. Hence, we can lower bound the above RHS as 
	\begin{align}				\label{eqn:comp-2}
		&\Ex_{u}\Pr_{(v_i,x_i,y_i)_{i = 1,2}}\Big[\one_S(v_1,\pi_{v_1,u}(x_1,y_1)) \neq \one_S(v_2,\pi_{v_2,u}(x_2,y_2))\Big]	\nonumber \\
		&\geq (1 - 2\eps)Pr_{(v_i,x_i,y_i)_{i = 1,2}}\Big[\one_S(v_1,\pi_{v_1,u}(x_1,y_1)) \neq \one_S(v_2,\pi_{v_2,u}(x_2,y_2))\Big].
	\end{align}
	Finally, note that conditioned on the event $\cE$, we have $\pi_{v_1,u}(\sigma(v_1)) = \pi_{v_2,u}(\sigma(v_2)) = \sigma(u)$, and hence, using the definition of the set $S$, we can simplify and lower bound the probability term as 
	\begin{align}		
	\Pr_{(v_i,x_i,y_i)_{i = 1,2}}\Big[\one_S(v_1,\pi_{v_1,u}(x_1,y_1)) \neq \one_S(v_2,\pi_{v_2,u}(x_2,y_2))\Big]
	&= \Pr_{(v_i,x_i,y_i)_{i = 1,2}}\Big[x_1(\sigma(u)) \neq x_2(\sigma(u))\Big]	\nonumber	\\
	&= \frac{1 - \rho^*}{2},			\label{eqn:comp-3}
	\end{align}
	where in the last step, we use the observation that $(x_1,x_2)$ are coordinate-wise $\rho^*$-correlated. Putting together \eqref{eqn:comp-1}, \eqref{eqn:comp-2}, and \eqref{eqn:comp-3}, we get that $S$ cuts at least $\left(\frac{1 - \rho^*}{2} - 2\eps\right)$-fraction of edges in $\cH$.
	
	For the second item, let us consider the dictator cut defined w.r.t. the $y$-variables i.e.,
	\[
	I := \Big\{(v,x,y) \in V_{\cH}~\Big|~v \in V' \text{ and } y(\sigma(v)) = 1 \Big\}
	\]
	Now we argue that $I$ indeed forms an independent set in $\cH$. To see this, for any right vertex, let $I_v := \{(v,x,y): (x,y) \in \Omega^R\}$ be the restriction of the independent set to vertex $v$ in $\cH$, and hence $I = \cup_{v \in V'} I_v$. Now fix any pair of vertices $v_1,v_2 \in V'$, and consider the following two cases.
	
	{\bf Case 1}. Suppose $v_1,v_2$ don't share a left-neighbor in $\cG$. Then, automatically, no edges go across $I_{v_1}$ and $I_{v_2}$ and hence 
	$w_{\cH}(I_{v_1},I_{v_2}) = 0$. 
	
	{\bf Case 2}. Suppose $v_1,v_2$ share at least one left-neighbor in $\cG$ (this also includes the case $v_1 = v_2$). Fix any such shared left-neighbor $u \in U$. Then, note that the weight of edges going between $I_{v_1}$ and $I_{v_2}$ relative to the weight of the edges introduced by the choice of $(u,v_1,v_2)$ in the dictatorship test is exactly
	\begin{align*}
	\Pr_{(x_1,y_1,x_2,y_2)}\left[y_1(\pi_{v_1,u}(\sigma(v_1))) = y_2(\pi_{v_2,u}(\sigma(v_2))) = 1\right]
	= \Pr_{(x_1,y_1,x_2,y_2)}\left[y_1(\sigma(u)) = y_2(\sigma(u)) = 1\right]
	= 0
	\end{align*} 
	where the first equality is due to $\pi_{v_1,u}(\sigma(v_1)) = \pi_{v_2,u}(\sigma(v_2)) = \sigma(u)$ (since $v_1,v_2 \in V'$), and the second equality is due to  $\mu_{2,\alpha}(1,1) = 0$.
	
	Combining the two cases, we can conclude that $w_\cH(I,I) = \sum_{v_1,v_2 \in V'} w_{\cH}(I_{v_1},I_{v_2}) = 0$ which implies that $I$ is an independent set. Furthermore, using arguments identical to those of the first item, we can show $I$ has weight at least $\alpha - 2\eps$.
\end{proof}

\subsection{Correlation Analysis}

Before we proceed to the soundness analysis of the reduction, we will need to establish some useful properties of the correlation structure of the distribution $\mu$. Recall from Section \ref{sec:corr-def} that given a joint probability space $(\Omega_1 \times \Omega_2,\mu)$, its correlation is defined as 
\[
\rho(\Omega_1,\Omega_2;\mu) = \max_{\substack{f,g:\Omega \to \mathbbm{R} \\ \Ex f = \Ex g = 0 \\ \|f\|_2,\|g\|_2 \leq 1}} \Ex_{(\bx,\by) \sim \mu} \Big[f(\bx)g(\by)\Big].
\]

For our analysis, we shall instantiate $\Omega := \{0,1\}^2$ and we let $\mu:= \mu_1 \otimes \mu_{2,\alpha}$ denote the joint distribution over $\{0,1\}^2 \times \{0,1\}^2$ corresponding to the test, and we shall need the following bound for the correlation of $(\Omega^2,\mu)$.

\begin{lemma}			\label{lem:corr}
	For every $\alpha \in [0,1]$ we have
	\[
	\rho(\Omega^2,\mu) = \max\{|\rho^*|,|\rho_{\alpha}|\}.
	\]
	In particular, when $\alpha \leq \alpha^*$, we have $\rho(\Omega^2,\mu) \leq |\rho^*|$.
\end{lemma}
\begin{proof}
	Recall that by definition,
	\begin{equation}			\label{eqn:corr-def}
		\rho(\Omega^2,\mu) = \max_{\substack{\Ex f = \Ex g = 0  \\ \|f\|_2,\|g\|_2 = 1}} \Ex_{(\bx,\by) \sim \mu}\big[f(\bx)g(\by)\big].
	\end{equation}
	Let $\{\phi_0 \equiv 1,\phi_1\}$ be a Fourier basis for $L_2(\mu_1)$ and let $\{\gamma_0 \equiv 1, \gamma_1\}$ be a Fourier basis for $L_2(\mu_2)$. Then note that using Lemma \ref{lem:eigval} we have
	\begin{equation}			\label{eqn:eigval}
	\Ex_{(\bx,\by)}\left[\phi_1(x_1)\phi_1(y_1)\right] = \rho^*
	\qquad\text{and}\qquad
	\Ex_{(\bx,\by)}\left[\gamma_1(x_2)\gamma_1(y_2)\right] = \rho_{\alpha}.
	\end{equation}
	Furthermore, $\{\phi_{a}\cdot\gamma_{b}\}_{a,b \in \{0,1\}}$ is a Fourier basis for $L_2(\mu_1 \otimes \mu_2)$. Hence for any feasible choice of $f,g$ (as in \eqref{eqn:corr-def}), we can write:
	\begin{align*}
		&\Ex_{(\bx,\by) \sim \mu} \left[f(\bx)g(\by)\right] \\
		&= \Ex_{(\bx,\by) \sim \mu}\Bigg[ \bigg(\wh{f}(0,0) + \wh{f}(1,0)\phi_1(x_1) + \wh{f}(0,1)\gamma_1(x_2) + \wh{f}(1,1) \phi_1(x_1)\gamma_1(x_2)\bigg) \\[-5pt]
		&  \qquad\qquad\qquad \times \bigg(\wh{g}(0,0) + \wh{g}(1,0)\phi_1(y_1) +\wh{g}(0,1)\gamma_1(y_2) + \wh{g}(1,1) \phi_1(y_1)\gamma_1(y_2)\bigg)\Bigg] \\[-5pt]
		&\overset{1}{=} \wh{f}(1,0)\wh{g}(1,0)\Ex_{(x_1,y_1)}\Big[\phi_1(x_1)\phi_1(y_1) \Big]+ \wh{f}(0,1)\wh{g}(0,1)\Ex_{(x_2,y_2)} \Big[\gamma_1(x_2)\gamma_1(y_2)\Big] \\  
		&  \qquad\qquad\qquad + \wh{f}(1,1)\wh{g}(1,1)\Ex_{(x_1,x_2,y_1,y_2)}\Big[\phi_1(x_1)\phi_1(y_1)\gamma_1(x_2)\gamma_1(y_2)\Big] \\
		&\overset{2}{=} \wh{f}(1,0)\wh{g}(1,0) \cdot \rho^* + \wh{f}(0,1)\wh{g}(0,1)\cdot\rho_\alpha + \wh{f}(1,1)\wh{g}(1,1)\cdot \rho^* \cdot \rho_{\alpha} \\[6pt]
		&\leq \max\{|\rho^*|,|\rho_\alpha|\} \sum_{(a,b) \in \{0,1\}^2} |\wh{f}(a,b)||\wh{g}(a,b)| \\[-6pt]
		&\overset{3}{\leq} \max\{|\rho^*|,|\rho_\alpha|\} \cdot \|f\|_2\|g\|_2 \\
		&\leq \max\{|\rho^*|,|\rho_\alpha|\},
	\end{align*}
	where in the above computation, in step $1$ we use the fact that $\wh{f}(0,0) = \wh{g}(0,0) = 0$, step $2$ follows due to \eqref{eqn:eigval}, and step $3$ can be argued as follows:
	\[
	\sum_{(a,b) \in \{0,1\}^2} |\wh{f}(a,b)||\wh{g}(a,b)| 
	\leq \sqrt{\sum_{(a,b) \in \{0,1\}^2} \wh{f}(a,b)^2}\sqrt{\sum_{(a,b) \in \{0,1\}^2} \wh{g}(a,b)^2}
	= \|f\|_2\|g\|_2,
	\]
	where again the first step is using Cauchy-Schwarz, and the second step is using Parseval's identity. This establishes the first part of the lemma. The second part follows immediately by combining the bound from the first part with Corollary \ref{corr:dist}.
\end{proof}

\subsection{Proof of Lemma \ref{lem:sound-iset}}

Suppose $\cG$ is a NO instance, i.e., ${\sf val}(\cG) \leq \delta$. Let $S \subseteq V_{\cH}$ be a subset in $\cH$, and let $f:V_\cH \to \{0,1\}$ be the indicator function. Furthermore, recall that $V_{\cH} := V \times \Omega^R = V \times \Omega^R$, where $\Omega = \{0,1\}_{1/2} \times \{0,1\}_{\alpha/2}$. Hence, for every $v \in V$, we can define the function $f_v:\Omega^R \to \{0,1\}$, which denotes the assignment to the code table corresponding to $v \in V$, i.e, $f_v(x,y) \defeq f(v,x,y)$ for all $(x,y) \in \Omega^R$. Now observe that for any Boolean valued variables $a,b$ we have
\[
\one\big(a \neq b\big) = a + b - 2ab.
\]
Using the above observation, we can express the fraction of edges cut by the assignment $f$ in $\cH$ as:
\begin{align}
	&\Pr_{(\omega_1,\omega_2) \sim E_{\cH}}\Big[f(\omega_1) \neq f(\omega_2)\Big] 	\non\\
	&= \Ex_{u \sim U}\Ex_{v_1,v_2 \sim N_{\cG}(u)}\Ex_{(x_1,y_1,x_2,y_2)}\Big[\one\Big(f(v_1,\pi_{v_1,u}(x_1,y_1)) \neq f(v_2,\pi_{v_2,u}(x_2,y_2))\Big)\Big]  \non \\
	&= \Ex_{u \sim U}\Ex_{v_1,v_2 \sim N_{\cG}(u)}\Ex_{(x_1,y_1,x_2,y_2)}
	\left[\one\Big(f_{v_1}\Big(\pi_{v_1,u}(x_1,y_1)\Big) \neq f_{v_2}\Big(\pi_{v_2,u}(x_2,y_2)\Big)\Big)\right] 		\non\\
	&= \Ex_{u \sim U}\Ex_{v_1,v_2 \sim N_{\cG}(u)}\Ex_{(x_1,y_1,x_2,y_2)}
	\Big[f_{v_1}\Big(\pi_{v_1,u}(x_1,y_1)\Big) + f_{v_2}\Big(\pi_{v_2,u}(x_2,y_2)\Big)				\non \\ 
	&\qquad\qquad\qquad\qquad 
	-2f_{v_1}\Big(\pi_{v_1,u}(x_1,y_1)\Big)f_{v_2}\Big(\pi_{v_2,u}(x_2,y_2)\Big)\Big] 		\non\\
	&= \Ex_{u \sim U}\Ex_{(x_1,y_1,x_2,y_2)}
	\Big[g_u(x_1,y_1) + g_u(x_2,y_2) - 2g_u(x_1,y_1)g_u(x_2,y_2)\Big],				\label{eqn:rhs}
\end{align}

where for every $u \in U$, the function $g_u:\Omega^R \to [0,1]$ is the averaged function defined as 
\[
g_u(x,y) \defeq \Ex_{v \sim N_G(u)}\Big[f_{v}\Big(\pi_{v,u}(x,y)\Big)\Big].
\]

\subsubsection{Bound For Non-Influential Vertices}

To further bound the RHS (i.e, Eq. \eqref{eqn:rhs}) from the above computation, we will need the following key technical lemma.

\begin{lemma}				\label{lem:struct}
	Let $h:\Omega^R \to [0,1]$ be a function satisfying the condition:
	\[
	\max_{j \in [R]}\Inf{j}{T_{1 - \eta}h} \leq \tau.
	\]
	Then, we have the bound:  
	\[
	\Ex_{(x_i,y_i)_{i = 1,2}} \Big[h(x_1,y_1) + h(x_2,y_2) - 2h(x_1,y_1)h(x_2,y_2)\Big]
	\leq 1 - 2\Lambda_{\rho^*}(1/2,1/2) + 2c(\tau,\eta).
	\]
\end{lemma}

\begin{proof}
	Throughout, let us denote $\omega_i := (x_i,y_i)$ for short. As a first step, we will find it more convenient to rewrite the term inside the expectation as
	\begin{align*}
		h(\omega_1) + h(\omega_2) - 2h(\omega_1)h(\omega_2)
		&= 1 - h(\omega_1)h(\omega_2) - \Big(1 - h(\omega_1)\Big)\Big(1 - h(\omega_2)\Big) \\
		&= 1 - h(\omega_1)h(\omega_2) - \bh(\omega_1)\bh(\omega_2),
	\end{align*}
	where $\bar{h}(\cdot) := 1 - h(\cdot)$. Hence, in terms of this notation, we want to bound,
	\begin{equation}				\label{eqn:rhs-1}
		1 - \Ex_{(\omega_1,\omega_2)}\Big[h(\omega_1)h(\omega_2) + \bh(\omega_1)\bh(\omega_2)\Big].
	\end{equation}
	Additionally, let us define $\delta_h := \Ex_{\omega \sim \Omega}[h(\omega)]$. Using Lemma \ref{lem:corr} and our choice of $\alpha$, we know that 
	\begin{equation}		\label{eqn:corr}
	\rho(\Omega^2;\mu) \leq \max\{|\rho_\alpha|,|\rho^*|\} = |\rho^*|.
	\end{equation}

	{\bf Introducing Implicit Noise}. Our first step is to introduce implicit noise into the expression. As observed in \eqref{eqn:corr}, note that we have $\rho(\Omega^2;\mu) \leq |\rho^*| < 1$. Hence, instantiating Lemma \ref{lem:noise} with $\Omega_1 = \Omega_2 = \Omega$, we get that 
	\begin{equation}		\label{eqn:noise}
	\Ex_{(\omega_1,\omega_2)}\left[h(\omega_1)h(\omega_2)\right]
	\geq \Ex_{(\omega_1,\omega_2)}\left[T_{1 - \eta}h(\omega_1)T_{1 - \eta}h(\omega_2)\right] - c'(\eta),
	\end{equation}
	where $c'(\eta) = c'(\eta,|\rho^*|)$ is the error-term from Lemma \ref{lem:noise}.

	{\bf Applying Noise Stability Bound}. Now we shall use Theorem \ref{thm:stab} to bound the above RHS. Again, from \eqref{eqn:corr} we get that $\rho(\Omega^2;\mu) \leq |\rho^*|$, and in the setting of the lemma we have 
	\[
	\max_{j \in [R]}\Inf{j}{T_{1 - \eta}h} \leq \tau.
	\]
	Hence, the probability space $(\Omega^2;\mu)$ and the function $f:= T_{1 - \eta}h$ satisfies the premise of Theorem \ref{thm:stab}, and therefore applying Theorem \ref{thm:stab} we get the bound:
	\begin{align}			
		\Ex_{(\omega_1,\omega_2)}\big[T_{1 - \eta}h(\omega_1)T_{1 - \eta}h(\omega_2)\big] 
		&\overset{1}{\geq} \Lambda_{-|\rho^*|}\big(\Ex[h], \Ex[h]\big) - c(\tau)			\non\\
		&= \Lambda_{-|\rho^*|}\big(\delta_h,\delta_h\big) - c(\tau)				\non\\		
		&= \Lambda_{\rho^*}\big(\delta_h,\delta_h\big) - c(\tau).				\label{eqn:st-1}
	\end{align}
	where in step $1$, we also used the observation that $\Ex_{\omega}[T_{1 - \eta}h(\omega)] = \Ex_{\omega}[h(\omega)]$. Putting together \eqref{eqn:noise} and \eqref{eqn:st-1}, we get the overall bound
	\begin{equation}			\label{eqn:st-2}
		\Ex_{(\omega_1,\omega_2)}\big[h(\omega_1)h(\omega_2)\big] \geq \Lambda_{\rho^*}\big(\delta_h,\delta_h\big) - c(\tau) - c'(\eta).
	\end{equation}
	
	Repeating the above steps with the functions $\bh = 1 - h$, and we also get the bound,
	\begin{align}
		\Ex_{(\omega_1,\omega_2)}\Big[\bh(\omega_1)\bh(\omega_2)\Big] 
		&\geq \Lambda_{\rho^*}\big(\Ex[\bh], \Ex[\bh]\big) - c(\tau)		\non\\
		&= \Lambda_{\rho^*}\big(1 - \delta_h,1 - \delta_h\big) - c(\tau) - c'(\eta).			\label{eqn:st-3}
	\end{align}
	Plugging in the bounds from \eqref{eqn:st-2} and \eqref{eqn:st-3} into \eqref{eqn:rhs-1} we now get that:
	\begin{align*}
		\Ex_{(\omega_1,\omega_2)}\Big[h(\omega_1)h(\omega_2) + \bh(\omega_1)\bh(\omega_2)\Big] 
		&\geq \Lambda_{\rho^*}\big(\Ex[h],\Ex[h]\big) + \Lambda_{\rho^*}\big(\Ex[\bh],\Ex[\bh]\big) - 2c(\tau) - 2c'(\eta) \\
		&= \Lambda_{\rho^*}\big(\delta_h,\delta_h\big) + \Lambda_{\rho^*}\big(1 - \delta_h,1 - \delta_h\big) - 2c(\tau) - 2c'(\eta) \\
		&\geq 2\Lambda_{\rho^*}\left(\frac12,\frac12\right) - 2c(\tau) - 2c'(\eta),
	\end{align*}
	where the last step follows using the convexity of the map $x \mapsto \Lambda_{\rho^*}(x,x)$ (Fact \ref{fact:convex}). Plugging in the above bound into \eqref{eqn:rhs-1} gives us the bound, and denoting $c(\tau,\eta):= 2c(\tau) + 2c'(\eta)$, concludes the proof of the lemma.
\end{proof}

\subsubsection{Upper Bounding \eqref{eqn:rhs}}

Let us now continue with bounding \eqref{eqn:rhs}. To that end, let us define the set
\[
U_{\rm good} := \left\{u \in U~\Big|~ \max_{j \in [R]}\Inf{j}{T_{1 - \eta}g_u} \leq \tau\right\}.
\]
The following lemma is a standard step in the soundness analysis which states that if $\cG$ is a NO instance, then only a tiny fraction of vertices don't belong to $U_{\rm good}$.

\begin{lemma}		\label{lem:inf-dec1}
	If ${\sf val}(\cG) \leq \delta$, then $\Pr_{u \sim U}\Big[u \in U_{\rm good}\Big] \geq 1 - c_1(\delta,\tau,\eta)$.
\end{lemma}

We defer the proof of the above lemma to Appendix \ref{sec:dec} for now, and continue with the soundness analysis. We can further bound:
\begin{align*}
	&\Ex_{u \sim U}\Ex_{(x_i,y_i)_{i = 1,2}}\Big[g_u(x_1,y_1) + g_u(x_2,y_2) - 2g_u(x_1,y_1)g_u(x_2,y_2)\Big] \\
	&\leq \Ex_{u \sim U_{\rm good}} \Ex_{(x_i,y_i)}\Big[g_u(x_1,y_1) + g_u(x_2,y_2) - 2g_u(x_1,y_1)g_u(x_2,y_2)\Big] + c_1(\delta,\tau,\eta) \\
	&\leq 1 - 2\Lambda_{\rho^*}(1/2,1/2) + 2c(\tau,\eta) + c_1(\delta,\tau,\eta) \\
	&\leq \frac{\arccos \rho^*}{\pi} + 2\eta 
\end{align*}	
where the first step is due to Lemma \ref{lem:inf-dec1} and the second step is due to Lemma \ref{lem:struct}, and in the third step, the second term is bounded by $\eta$ due to our setting of parameters. Combining the above bound with \eqref{eqn:rhs}, we get that the
\begin{align*}
	\Pr_{e \sim \cH}\Big[\text{$f$ cuts $e$}\Big]
	&\leq \Ex_{u \sim U}\Ex_{(x_1,y_1,x_2,y_2)}
	\Big[g_u(x_1,y_1) + g_u(x_2,y_2) - 2g_u(x_1,y_1)g_u(x_2,y_2)\Big] \\
	&\leq \frac{\arccos \rho^*}{\pi} + 2 \eta.
\end{align*}
Since the above bound holds for any subset $S \subseteq V_{\cH}$, this concludes the proof of Lemma \ref{lem:sound-iset}.

\section{Proof of Theorem~\ref{thm:algo}}

    Let $G = (V, E)$ be a graph such that there exists an independent set $I \subseteq V$ satisfying $w(I,I^c) = \alpha > \alpha^*$. We aim to give an algorithm that produces a cut $(S, S^c)$ such that $w(S, S^c) > \frac{\arccos \rho^\ast}{\pi}$. 
    Consider the following integer program:
    \[
    \begin{array}{ll}
        \text{max} & \frac{1}{|E|}\sum_{\{i, j\} \in E} \frac{1 - X_iX_j}{2} \\
        \text{s.t.} & \forall \{i, j\} \in E,\, (1 - X_i)(1 - X_j) = 0 \\
        & \forall i \in V,\, X_i \in \{1, -1\}. 
    \end{array}
    \]
    A solution to the above program is feasible if and only $\{i \mid X_i = -1\}$ is an independent set. By our assumption on $G$, if we take $X_i = -1$ for $i \in I$ and $X_i = 1$ for $i \not\in I$, then $\sum_{\{i, j\} \in E} \frac{1 - X_iX_j}{2} \geq \alpha |E|$, so the optimum of this integer program is at least $\alpha$. We now relax this program to a semi-definite program in the natural way:
    \begin{figure}[h!]
    \[
    \begin{array}{ll}
        \text{max} & \frac{1}{|E|}\sum_{\{i, j\} \in E} \frac{1 - \bv_i \cdot \bv_j}{2} \\
        \text{s.t.} & \forall \{i, j\} \in E,\, \langle \bv_0 - \bv_i,\bv_0 - \bv_j\rangle = 0 \\
        & \forall i \in \{0\} \cup V,\, \|\bv_i\|^2 = 1. 
    \end{array}
    \]
    \caption{The semi-definite relaxation $\text{SDP}(G)$}\label{fig:sdp}
    \end{figure}

    The relaxation can only increase the objective value, so we can find in polynomial time a set of vectors $\{\bv_i \mid i \in \{0\} \cup V\}$ satisfying the SDP constraints such that $\frac{1}{|E|}\sum_{\{i, j\} \in E} \frac{1 - \bv_i \cdot \bv_j}{2} \geq \alpha$. We will give a rounding algorithm that obtains a large cut given such an SDP solution.
    
    For every $i,j \in V$, let us define $b_i = \bv_i \cdot \bv_0$ and $b_{ij} = \bv_i \cdot \bv_j$. It follows from the SDP constraints that for every edge $\{i, j\}$, $b_{ij} = -1 + b_i + b_j$, and therefore $b_i + b_j \geq 0$. 

    \begin{definition}\label{def:rho}
        For every $b_i, b_j \in [-1, 1]$, we define $\rho(b_i, b_j) = -\sqrt{\frac{(1 - b_i)(1 - b_j)}{(1 + b_i)(1 + b_j)}}$ if $(1 + b_i)(1 + b_j) \neq 0$, and 0 otherwise.
    \end{definition}

     For any $\rho \in [-1, 1]$, let $\varphi_\rho$ and $\Phi_\rho$ be the p.d.f. and c.d.f. of the 2-dimensional normal distribution with covariance matrix $\left(\begin{array}{cc}
        1 & \rho \\
        \rho & 1
    \end{array}\right)$. 

    The rounding algorithm that we use comes from the $\THRESH^-$ family, proposed by Lewin, Livnat, and Zwick~\cite{lewin2002improved}.
    \begin{definition}[$\THRESH^-$]
        Let $t: [-1, 1] \to [0, 1]$ be some continuous function. The $\THRESH^-$ rounding scheme specified by $t$ works as follows:
        \begin{itemize}
            \item Choose a standard Gaussian vector $\br$.
            \item For every $i$, include $i \in S$ if $\br \cdot \bv_i^\perp \geq \Phi^{-1}(t(b_i))$, and $i \in V \setminus S$ otherwise. Here $\bv_i^\perp = \frac{\bv_i - (\bv_i \cdot \bv_0) \bv_0}{\|\bv_i - (\bv_i \cdot \bv_0) \bv_0\|}$.
        \end{itemize} 
    \end{definition}
    Intuitively, a $\THRESH^-$ scheme generalizes Goemans-Williamson hyperplane rounding in two important ways: (1) we operate on $\bv_i^\perp$ instead of $\bv_i$; and (2) the hyperplane we choose may be shifted away from the origin and the amount of shifting (specified by $t$) depends on $b_i$.

    \begin{proposition}[Soundness]\label{prop:sound}
        The soundness on the edge $(b_i, b_j)$ achieved by the $\THRESH^-$ rounding scheme with threshold function $t: [-1, 1] \to [0, 1]$ is equal to
        \begin{equation}\label{eq:soundness}
            s(b_i, b_j) = 1 - \Phi_\rho(\Phi^{-1}(t(b_i)), \Phi^{-1}(t(b_j))) - \Phi_\rho(\Phi^{-1}(1 - t(b_i)), \Phi^{-1}(1 - t(b_j)))
        \end{equation}
        where $\rho = \rho(b_i, b_j)$.
    \end{proposition}
    \begin{proof}
        For every $\{i, j\} \in E$, the probability that $i$ and $j$ do not land on the same side of $(S, V \setminus S)$ is exactly the right hand side of \eqref{eq:soundness}.
    \end{proof}


    \begin{lemma}\label{lem:affine_sound}
        Let $\frac{1}{|E|}\sum_{\{i, j\} \in E} \frac{1 - b_{ij}}{2} = \alpha$, and let $b_0 = 1 - \alpha$. Assume that there is some $\THRESH^-$ rounding scheme $t:[-1, 1] \to [0, 1]$ whose soundness $s(b_i, b_j)$ satisfies the following properties:
        \begin{itemize}
            \item $s(b_i, b_j) \geq p \frac{b_i + b_j}{2} + q$ for every $b_i, b_j$ satisfying $b_i + b_j \geq 0$,
            \item $s(b_0, b_0) = p\cdot b_0 + q$,
            \item $t(b_0) = 1/2$.
        \end{itemize}
        Then 
        \[
            \E[w(S, V\setminus S)] \geq \frac{\arccos(\rho(b_0, b_0))}{\pi}.
        \]
        In particular, if in addition we also have $\alpha > \alpha^\ast$, then $\E[w(S, V\setminus S)] > \frac{\arccos(\rho^\ast)}{\pi}$.
    \end{lemma}
    \begin{proof}
        We have 
        \begin{align*}
            \E[w(S, V\setminus S)] & = \frac{1}{|E|}\sum_{\{i, j\} \in E} s(b_i, b_j) \\ 
            & \geq \frac{1}{|E|}\sum_{\{i, j\} \in E} \left( p \cdot \frac{b_i + b_j}{2} + q \right)\\
            & =\frac{1}{|E|} \sum_{\{i, j\} \in E} \left( p \cdot \frac{1 + b_{ij}}{2} + q \right)\\
            & = p\cdot b_0 + q = s(b_0, b_0). 
        \end{align*}
        On the other hand, by Proposition~\ref{prop:sound}, since $t(b_0) = 1/2$, we have $s(b_0, b_0) = \frac{\arccos(\rho(b_0, b_0))}{\pi}$.
    \end{proof}
    
    By the above lemma, it is sufficient to find some $t: [-1, 1] \to [0, 1]$ and some $p, q \in \mathbb{R}$ such that $s(b_i, b_j) \geq p \frac{b_i + b_j}{2} + q$ for every feasible $(b_i, b_j)$. Furthermore, if $\frac{1}{|E|}\sum_{\{i, j\} \in E} \frac{1 - b_{ij}}{2} = \alpha$ and $b_0 = 1 - \alpha$, then we also want $t(b_0) = 1/2$, since in this case we may have $b_i = b_0$ for every $i \in V$ and the best thing to do when this happens is to use a hyperplane centered at the origin. It turns out that the choice of such $p$ and $q$ is fixed under these mild assumptions. 

    \begin{proposition}\label{proposition:pq}
        Let $b_0 \in (0, 1)$ and $s_0 = \frac{\arccos(\rho(b_0, b_0))}{\pi}$. Let $t: [-1, 1] \to [0, 1]$ be some function that is differentiable at $b_0$ with $t(b_0) = 1/2$. Let $p, q \in \mathbb{R}$ be two real numbers such that the following holds:
        \begin{itemize}
            \item $s(b_i, b_j) \geq p \cdot \frac{b_i + b_j}{2} + q$ for every $b_i, b_j \in [-1, 1]$,
            \item $s_0 = p \cdot b_0 + q$.
        \end{itemize}
        Then we have
        \begin{equation}\label{eq:pq}
            p = - \frac{2}{(1 + b_0)^2 \cdot \pi \cdot \sin(\pi \cdot s_0)}, \quad q = -p\cdot b_0 + s_0.
        \end{equation}
    \end{proposition}
    \begin{proof}
        Let $s(b) \coloneqq s(b, b)$. By differentiability of $t$ at $b_0$ and Lemma~\ref{lem:partial_s}, we have
        \begin{equation*}
            \left.\frac{\partial}{\partial b}s(b) \right|_{b = b_0}= -2 \cdot \varphi_\rho(0, 0) \cdot \left.\frac{\partial}{\partial b}\rho(b, b) \right|_{b = b_0} = -2 \cdot \frac{1}{2\pi \sqrt{1 - \rho(b_0, b_0)^2}} \cdot \frac{2}{(1+b_0)^2}.
        \end{equation*}
        To satisfy the given conditions, it must be the case that
        \begin{equation*}
            p = \left.\frac{\partial}{\partial b}s(b) \right|_{b = b_0} = - \frac{2}{(1 + b_0)^2 \cdot \pi \cdot \sin(\pi \cdot s_0)}.
        \end{equation*}
        Here in the last equality we used the fact that $\cos(\pi\cdot s_0) = \rho(b_0, b_0)$ and $s_0 \in (0, 1)$. 
    \end{proof}

    For the choice $b_0 = b^\ast = 0.1840220...$ (which corresponds to $\alpha = \alpha^\ast = 0.81597...$, we have $p = -0.6266938...$ and $q = 0.8573447...$. We now present a $\THRESH^-$ scheme for which the conditions for soundness in Proposition~\ref{proposition:pq} are achieved for every $b_0$ in a tiny neighborhood of $b^\ast$.

    Let us first describe the heuristics which we used to discover the $\THRESH^-$ scheme. We do so by looking at some special configurations.

    \begin{lemma}\label{lem:sound_-b_b}
        Consider the edge $(-b, b)$. We must have $q \leq t(b) + t(-b) \leq 2 - q$.
    \end{lemma}
    \begin{proof}
        Note that $\rho(-b, b) = -1$ for every $b \in [-1, 1]$. For every $x, y \in \mathbb{R}$, we have
        \begin{equation*}
            \Phi_{-1}(x, y) = \max(0, \Phi(x) + \Phi(y) - 1).
        \end{equation*}
        It follows that $q \leq s(-b, b) = 1 - \max(t(b) + t(-b) - 1, 0) - \max(1 - t(b) - t(-b), 0) = 1 - |t(b) + t(-b) - 1|$, which simplifies to $q \leq t(b) + t(-b) \leq 2 - q$.
    \end{proof}

    \begin{lemma}\label{lem:sound_b_1}
        Consider the edge $(b, 1)$. We must have $t(b) + t(1) - 2t(b) \cdot t(1)\geq p \cdot \frac{b+ 1}{2} + q$.
    \end{lemma}
    \begin{proof}
        Note that $\rho(b, 1) = 0$, and $\Phi_0(\Phi^{-1}(x), \Phi^{-1}(y)) = \Phi(\Phi^{-1}(x))\Phi(\Phi^{-1}(y)) = xy$ for every $x, y \in [-1, 1]$, so we have
        \begin{equation*}
            s(b, 1) = 1 - t(b) \cdot t(1) - (1 - t(b))\cdot(1 - t(1)) = t(b) + t(1) - 2t(b) \cdot t(1)\geq p \cdot \frac{b+ 1}{2} + q.
        \end{equation*}
    \end{proof}

We now construct a piecewise linear threshold function $t$ that adheres to the constraints in Lemmas~\ref{lem:sound_-b_b} and \ref{lem:sound_b_1}. We start with 5 control points for $t$: $-1$, $-b_0$, 0, $b_0$, 1. It is clear that we must have $t(b_0) = 1/2$. By setting $b = 1$ in Lemmas~\ref{lem:sound_-b_b} and \ref{lem:sound_b_1}, we obtain two equations:
\begin{align}
    q \leq t(1) + t(-1) \leq 2 - q \\
    t(1) - t(1)^2 \geq \frac{p + q}{2}. \label{eq:condition_t1}
\end{align}

It is desirable to make $t(1)$ as large as possible, as by Lemma~\ref{lem:sound_b_1} this helps edges of the form $(b, 1)$ with $b < b_0$. We therefore choose $t(-1) = 0$ and $t(1) = q + 0.0095$. Here, if we increase $t(1)$ further to $q + 0.01$, then \eqref{eq:condition_t1} will be violated.

For $t(-b_0)$ and $t(0)$, again by Lemma~\ref{lem:sound_-b_b}, we must have $t(-b_0) + t(b_0) = t(-b_0) + 1/2 \geq q $ and $2t(0) \geq q$, so we choose $t(-b_0) = q - 1/2 + \delta$ and $t(0) = q/2 + \delta$ for some small $\delta = 0.001$.

Finally, in order to make sure $t$ is smooth around a small neighborhood of $b_0$, we choose some very tiny $r$ such that $t$ has the same slope on $[b_0 - r, b_0 + r]$. The slope is chosen to be 0.42, so that the Hessian of the soundness in this neighborhood is positive definite.

Summarizing, the final scheme is presented in the following table showing its break points and the corresponding values:
\begin{table}[H]
\centering
\begin{tabular}{c|ccccccc}
   $b$     & $-1$ & $-b_0$ & 0 & $b_0 - r$ & $b_0$ & $b_0 + r$ & 1 \\ \hline
   $t(b)$  & 0 & $q - \frac{1}{2} + \delta$ & $q + \delta$ & $\frac{1}{2} - 0.42 r$ & $\frac{1}{2}$ & $\frac{1}{2} + 0.42 r$ & $q + 0.0095$
\end{tabular}
\caption{Specification of the threshold function $t$. $\delta = r = 0.001$.}\label{table:t}
\end{table}

\begin{figure}[H]
    \centering
    \includegraphics[width=0.5\linewidth]{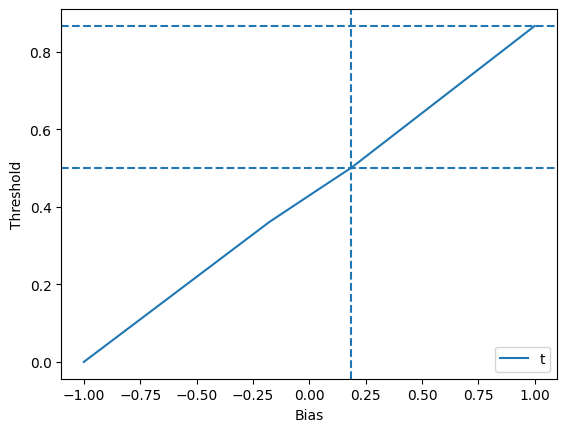}
    \caption{Plot of the threshold function $t$ as a function of the bias $b$. Plots in this paper are made with Matplotlib~\cite{Hunter:2007}.}
    \label{fig:thresh_plot}
\end{figure}

The verification of this scheme can be summarized in the following lemma.

\begin{lemma}[Interval Arithmetic]\label{lem:ia}\footnote{The interval arithmetic code can be found at \url{https://github.com/nenghuang/max-cut-independent-set}.}
    Let $\eta = 10^{-13}$ and $b_0 \in [b^\ast - \eta, b^\ast + \eta]$. Let $t$ be defined as in Table~\ref{table:t}, and let $p, q$ be as in~\eqref{eq:pq}. For any $b_i, b_j \in [-1, 1]$ with $b_i + b_j \geq 0$, we have
    \begin{equation}\label{eq:soundness_bound}
        s(b_i, b_j) - \left( p \cdot \frac{b_i + b_j}{2} + q \right) \geq 0.
    \end{equation}
\end{lemma}
\begin{proof}[Proof sketch]
    Let $r = 0.001$. Using interval arithmetic we verify two cases:
    \begin{itemize}
        \item $|b_i - b_0| \geq r$ or $|b_j - b_0| \geq r$. In this case \eqref{eq:soundness_bound} is a strict inequality and can be verified directly.
        \item $|b_i - b_0| \leq r$ and $|b_j - b_0| \leq r$. In this case, interval arithmetic cannot establish \eqref{eq:soundness_bound} directly since equality can be achieved. Instead, we prove that the Hessian of the LHS of~\eqref{eq:soundness_bound} (using Lemma~\ref{lem:partial_s}) is positive definite in this region. This implies that a critical point in this region must also be a global minimum. Since $(b_0, b_0)$ is a critical point (by Proposition~\ref{proposition:pq}), it follows that 
        \begin{equation*}
            s(b_i, b_j) - \left( p \cdot \frac{b_i + b_j}{2} + q \right) \geq s(b_0, b_0) - \left( p \cdot \frac{b_0 + b_0}{2} + q \right) = 0.
        \end{equation*}
    \end{itemize}. 
\end{proof}

\begin{figure}[H]
    \centering
    \includegraphics[width=0.5\linewidth]{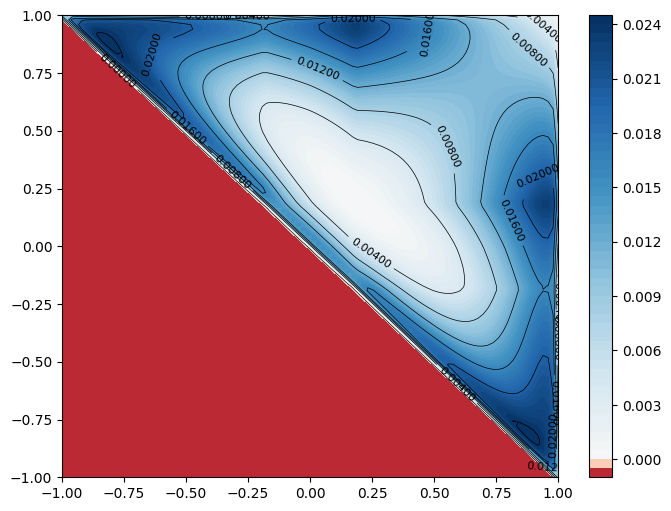}
    \caption{Contour plot of the function $s(b_i, b_j) - \left( p \cdot \frac{b_i + b_j}{2} + q \right)$. Red region is infeasible.}
    \label{fig:contour}
\end{figure}

\begin{remark}
    The verification code for the above lemma is based on the Arb library~\cite{johansson2017arb}, which is part of the FLINT C libraries~\cite{flint}. We made use of a C++ wrapper for a fragment of Arb library which was recently made for verifying approximation ratios for MAX 2-CSP~\cite{BHPZ23, BHZ24}.
\end{remark}

    We note that since in our verification we chose $b_0 \in [b^\ast - \eta, b^\ast + \eta]$ for some small $\eta > 0$, it shows that we can do slightly better than $s^\ast$ if the completeness is slightly above $\alpha^\ast$.

We are now ready to prove Theorem~\ref{thm:algo}.

\begin{proof}[Proof of Theorem~\ref{thm:algo}]
    Let $\bv_0, \bv_1, \ldots, \bv_n$ be an SDP solution to $\text{SDP}(G)$ such that $\frac{1}{|E|} \cdot \sum_{\{i, j\}\in E}\frac{1 - \bv_i \cdot \bv_j}{2} = \alpha > \alpha^*$. Note that the vectors $\bv'_0 = \bv'_1 = \cdots = \bv'_n = \bv_0$ is also a feasible solution to $\text{SDP}(G)$, which has objective value 0. So by taking a suitable combination between these two SDP solutions, we may assume that the objective value of our SDP solution is at most $\alpha^\ast + 10^{-13}$. We may then apply the $\THRESH^-$ scheme described in Table~\ref{table:t}, and by Lemma~\ref{lem:affine_sound} and Lemma~\ref{lem:ia} it produces a cut $(S, V\setminus S)$ such that $\E[w(S, V\setminus S)] > \frac{\arccos(\rho^\ast)}{\pi}$.
\end{proof}

Admittedly, the above proof is not entirely satisfactory in the sense that it decreased the SDP value artificially for the sake of the argument. It is an interesting question whether there is an algorithm that takes advantage of larger SDP values. We state this as the following conjecture.

\begin{conjecture}
    For every $\alpha \geq \alpha^*$, there exists an algorithm such that, given a graph $G = (V, E)$ with an independent set $I \subseteq V$ such that $w(I,I^c) \geq \alpha$,  finds a set $S \subset V$ in polynomial time such that $w(S,S^c) > \frac{\arccos \rho(\frac{1 - \alpha}{2}, \frac{1-\alpha}{2})}{\pi}$, where $\rho(b_i, b_j)$ is defined as in Definition~\ref{def:rho}.
\end{conjecture}

\section{Proof of Theorem \ref{thm:iset-2}}

We describe the gadget for the theorem in the figure below.

\begin{figure}[ht!]
	\begin{mdframed}
	Consider the weighted graph $\cH = (V_{\cH},E_{\cH},w_{\cH})$ defined in the following way.
		
	{\bf Vertices}. Let $V_1$ and $V_2$ be two copies of the hypercube:
	\[
		V_1 := \left\{(1,x) ~\big|~ x \in \{0,1\}^R\right\}
		\qquad\qquad\text{and}\qquad\qquad
		V_2 := \left\{(2,x)~\big|~x \in \{0,1\}^R\right\}
	\]
	We let the vertex set be $\cV_{\cH} = V_1 \cup V_2$. \\[-10pt]

	{\bf Edges}. The distribution $w_{\cH}$ over edges $E_{\cH}$ is defined using the following probabilistic process:
	
	\begin{enumerate}
		\item Sample $(x,y) \in \{0,1\}^R \times \{0,1\}^R$ such that $x$ and $y$ are $\rho^*$-correlated.
		\item W.p. $\alpha$, output the edge $\{(1,x),(2,y)\}$ and w.p. $1 - \alpha$ output the edge $\{(2,x),(2,y)\}$. \\[-10pt]
	\end{enumerate}

	In other words, w.p. $\alpha$ we output a $\rho^*$-correlated pair in $V_1 \times V_2$, and with probability $(1 - \alpha)$, we output a $\rho^*$-correlated pair in $V_2 \times V_2$.
	\end{mdframed}
	\caption{Gadget with Known Independent Set}
	\label{fig:gadget-2}
\end{figure} 

We begin by stating the completeness properties of the gadget that are straightforward to verify.

\begin{lemma}[Completeness]			\label{lem:comp-2}
	The graph $\cH$ defined in Figure \ref{fig:gadget-2} satisfies the following properties:
	\begin{itemize}
		\item The set $V_1$ is an independent set in $\cH$. Moreover $w_{\cH}(V_1,V_2) = \alpha$.
		\item Consider the set 
		\[
		S := \left\{(i,x) : x(1) = 1, i \in \{1,2\}\right\}.
		\]
		Then $w_{\cH}(S,S^c) \geq \frac{1 - \rho^*}{2}$.
	\end{itemize}
\end{lemma}
The above lemma can be proved exactly as Lemma \ref{lem:comp-iset} and we omit the details here. As before, the more technically involved step is to argue the soundness guarantee of the above gadget, stated in the lemma below.

\begin{lemma}		\label{lem:sound-2}
	Suppose $f:V_1 \cup V_2 \to \{0,1\}$ is an assignment satisfying
	\[
	\max_{a = 1,2}\max_{j \in [R]} \Inf{j}{f_a} \leq \tau.
	\]
	Then,
	\[
	\Ex_{(x,y) \sim \cH} \Big[{\sf cut}(f(x),f(y))\Big]
	\leq \max\left\{\frac{\arccos \rho^*}{\pi},\alpha\right\} + c(\tau),
	\]
	where $c(\tau)$ is an increasing function of $\tau$ satisfying $c(0) = 0$.
\end{lemma}

The following sections are dedicated to proving the above lemma.

\subsection{Proof of Lemma \ref{lem:sound-2}}

To begin with, let us consider an assignment $(f_1,f_2): V_1 \cup V_2 \to \{0,1\}$ such that 
\begin{equation}			\label{eqn:small-inf}
\max_{a \in \{1,2\}}\max_{j \in [R]} \Inf{j}{f_a} \leq \tau, 
\end{equation}
where $\tau \leq \tau_0$. Then, as in the proof of Lemma \ref{lem:sound-3col}, we can express the weight of the edges cut by the assignment $(f_1,f_2)$ as follows:
\[
\Pr_{e \sim \cH}\left[\text{$(f_1,f_2)$ cuts $e$}\right]
= \alpha \Ex_{x \underset{\rho^*}{\sim} y} \Big[{\sf Cut}(f_1(x),f_2(y))\Big]
+ (1 - \alpha) \Ex_{x \underset{\rho^*}{\sim} y} \Big[{\sf Cut}(f_2(x),f_2(y))\Big].
\]

Towards upper bounding the RHS, our first step is the following lemma, which bounds the objective in terms of the Gaussian cut functions.

\begin{lemma}			\label{lem:iset-2}
	Suppose the assignment $(f_1,f_2)$ satisfies \eqref{eqn:small-inf}. Then we can bound,
	\begin{equation}			\label{eqn:iset-2}
	\Pr_{e \sim \cH}\left[\text{$(f_1,f_2)$ cuts $e$}\right]
	\leq \max_{\mu_1,\mu_2 \in [0,1]} \alpha {\sf Cut}_\rho(\mu_1,\mu_2) + (1 - \alpha) {\sf Cut}_\rho(\mu_2,\mu_2) + c(\tau),
	\end{equation}
	for some continuous increasing function $\c(\tau)$ satisfying $c(0) = 0$.
\end{lemma}
The proof of the above lemma proceeds exactly as in the derivation of \eqref{eqn:tri-1} in the proof of Lemma \ref{lem:main}, and we omit it here. Now, we state the key lemma which gives closed-form solutions to the maximization problem in the RHS of \eqref{eqn:iset-2}. 

\begin{lemma}				\label{lem:iset-3}
	Let us denote $s^*  = \frac{\arccos \rho^*}{\pi}$. Then, the following holds:
	\[
	\max_{\mu_1,\mu_2 \in [0,1]} \alpha {\sf Cut}_\rho(\mu_1,\mu_2) + (1 - \alpha) {\sf Cut}_\rho(\mu_2,\mu_2)
	= 
	\begin{cases}
	\frac{\arccos \rho^*}{\pi}	& \text{ if } \alpha \leq s^*,  \\
	\alpha & \text{ if } \alpha > s^*				 	
	\end{cases}
	\]
\end{lemma}

The rest of this section proves the above lemma.

\subsection{Proof of Lemma \ref{lem:iset-3}}

For brevity, let us denote,
\[
F_{{\rm cut},\alpha}(\mu_1,\mu_2) = \alpha \ccut(\mu_1,\mu_2) + (1 - \alpha) \ccut(\mu_2,\mu_2).
\]
Again, as before, we will find it more convenient to work with 
\[
F_\alpha(\mu_1,\mu_2) = \alpha \ccut(\mu_1,\mu_2) + (1 - \alpha) \ccut(\mu_2,\mu_2).
\]
Let us first compute the expressions for the partial derivatives of $F_\alpha$.

\begin{lemma}				\label{lem:iset-dv}
	The following holds:
	\begin{itemize}
		\item 
		\[
		\frac{\partial F_{\alpha}}{\partial \mu_1}
		= \alpha G\left(\frac{t(\mu_2) - \rho t(\mu_1)}{\sqrt{1 - \rho^2}}\right).
		\]
		\item 
		\[
		\frac{\partial F_{\alpha}}{\partial \mu_2}
		= \alpha G\left(\frac{t(\mu_1) - \rho t(\mu_2)}{\sqrt{1 - \rho^2}}\right)
		+ 2(1 - \alpha)G\left(\frac{t(\mu_2) - \rho t(\mu_2)}{\sqrt{1 - \rho^2}}\right),
		\]
		where recall that $G(z) = \Phi(z) - \Phi(-z)$.
	\end{itemize}	
\end{lemma}
\begin{proof}
	
The first item is an immediate consequence of Lemma \ref{lem:dv}, and hence we will focus on proving the second item. We first note that by definition,
\begin{equation}			\label{eqn:iset-dv1}
\frac{\partial F_{\alpha}(\mu_1,\mu_2)}{\partial \mu_2} = \alpha \frac{\partial \ccut(\mu_1,\mu_2)}{\partial \mu_2} + (1 - \alpha) \frac{\partial \ccut(\mu_2,\mu_2)}{\partial \mu_2}.
\end{equation}	
Again, similar to the first item and Lemma \ref{lem:dv}, we get that
\begin{equation}			\label{eqn:iset-dv2}	
	\frac{\partial \ccut(\mu_1,\mu_2)}{\partial \mu_2 } = G\left(\frac{t(\mu_1) - \rho t(\mu_2)}{\sqrt{1 - \rho^2}}\right).
\end{equation}	
The derivative of the second term in \eqref{eqn:iset-dv1} needs to be handled slightly differently. Let us define $x_\mu = y_\mu = \mu$. Then, again using the chain rule of derivatives, we get that
\begin{align}
	\frac{\partial \ccut(x_\mu,y_\mu)}{\mu}
	&= \frac{\partial x_\mu}{\partial \mu}\cdot\frac{\partial \ccut(x_\mu,y_\mu)}{\partial x_\mu}
	+ \frac{\partial y_\mu}{\partial \mu}\cdot\frac{\partial \ccut(x_\mu,y_\mu)}{\partial y_\mu}		\non\\
	&= G\left(\frac{t(y_{\mu}) - \rho t(x_{\mu})}{\sqrt{1 - \rho^2}}\right) 
	+ G\left(\frac{t(x_\mu) - \rho t(y_\mu)}{\sqrt{1 - \rho^2}}\right) 		\non\\
	&= 2 G\left(\frac{t(\mu) - \rho t(\mu) }{\sqrt{1 - \rho^2}}\right).		\label{eqn:iset-dv3}
\end{align}	
Hence plugging in the expressions from \eqref{eqn:iset-dv2} and \eqref{eqn:iset-dv3} into \eqref{eqn:iset-dv1} we obtain the expression:
\[
\frac{\partial \ccut(\mu_1,\mu_2)}{\partial \mu_2}
= \alpha G\left(\frac{t(\mu_1) - \rho t(\mu_2)}{\sqrt{1 - \rho^2}}\right)
+ 2(1 - \alpha) G \left(\frac{t(\mu_2) - \rho t(\mu_2)}{\sqrt{1 - \rho^2}}\right).
\]
\end{proof}

\subsection{Interior Analysis}

Suppose $(\mu_1,\mu_2)$ is an interior point of $F_{\alpha}$. Then, using the first order condition w.r.t. $\mu_1$ and Lemma \ref{lem:iset-dv}, we get that
\[
0 = \frac{\partial F_{\alpha}}{\partial \mu_1} = 
\Phi\left(\frac{t(\mu_2) - \rho t(\mu_1)}{\sqrt{1 - \rho^2}}\right) - \Phi\left(- \frac{t(\mu_2) - \rho t(\mu_1)}{\sqrt{1 - \rho^2}}\right).
\]
As noted before, the function $z \mapsto \Phi(z) - \Phi(-z)$ is odd and strictly increasing, and hence, the above constraint implies that
\[
\frac{t(\mu_2) - \rho t(\mu_1)}{\sqrt{1 - \rho^2}} = 0,
\]
or equivalently, $t(\mu_1) = t(\mu_2)/\rho$. Therefore, substituting this into the first order condition w.r.t. $\mu_2$, we get the constraint 
\begin{align*}
0 
= \frac{\partial F_{\alpha}(\mu_1,\mu_2)}{\partial \mu_2}
&= \alpha G\left(\frac{t(\mu_1) - \rho t(\mu_2)}{\sqrt{1 - \rho^2}}\right) + 2(1 - \alpha)G\left(\frac{t(\mu_2) - \rho t(\mu_2)}{\sqrt{1 - \rho^2}}\right) \\
&= \alpha G\left(t(\mu_2)\frac{\sqrt{1 - \rho^2}}{\rho}\right) + 2(1 - \alpha)G\left(t(\mu_2)\sqrt{\frac{1 - \rho}{1 + \rho}}\right) \\
&= \alpha G \Big(A t(\mu)\Big) + 2(1 - \alpha) G(B t(\mu)),
\end{align*}
where 
\[
A = \frac{\sqrt{1 - \rho^2}}{\rho} \approx -1.051
\qquad\qquad\text{and}\qquad\qquad
B = \sqrt{\frac{1 - \rho}{1 + \rho}} \approx 2.33109,
\]
when $\rho = \rho^*$. Now, we have the following main lemma.
\begin{lemma}			\label{lem:iset-dv1}
	Assume $\rho \leq 0$. Then, subject to the constraint $t(\mu_2) = \rho t(\mu_1)$, the only local minima of $F_\alpha$ is at $\mu_2 = 1/2$. Hence, $(1/2,1/2)$ is the local minima of $F_{\alpha}$. 
\end{lemma}

The above lemma characterizes the local minima of $F_{\alpha}$ in the interior of $F$. We defer the proof of the above lemma to Section \ref{sec:iset-dv1}, and proceed with the analysis. To finish the proof, we consider the following cases. Let $\mu^* = (\mu^*_1,\mu^*_2)$ be the global minimizer of $F_\alpha$. Then,

{\bf Case (i)}. Suppose $\mu^*$ is in the interior of $[0,1]^2$. Then $\mu^*$ must be a local maximum of $F_{{\rm cut},\alpha}$ and hence from Lemma \ref{lem:iset-dv1}, we get that $(\mu^*_1,\mu^*_2) = (1/2,1/2)$, and the corresponding cut value is
\[
F_{{\rm cut},\alpha}(1/2,1/2) = {\sf Cut}_{\rho^*}(1/2,1/2) = \frac{\arccos \rho^*}{\pi}.
\]

{\bf Case (ii)}. Suppose $(\mu^*_1,\mu^*_2)$ is at the surface of $[0,1]^2$ i.e., $\mu^*_1 \in \{0,1\}$ or $\mu^*_2 \in \{0,1\}$. Then considering all the cases, we get that the maximizers of $F_{{\rm cut},\alpha}$ are of the form $(0,1)$ and $(1,0)$, with the cut value $F_{{\rm cut},\alpha}(1,0) = F_{{\rm cut},\alpha}(0,1) = \alpha$.

Combining the two cases, we get that
\[
F_{{\rm cut},\alpha}(\mu^*_1,\mu^*_2) = \max\left\{\frac{\arccos \rho^*}{\pi}, \alpha\right\}.
\]

\subsection{Proof of Theorem \ref{thm:iset-2}}

The actual reduction for the proof of Theorem \ref{thm:iset-2} follows similarly to that of Theorem \ref{thm:3-col}, where we replace the tripartite noisy-cube gadget with the gadget in Figure \ref{fig:gadget-2}. We describe the reduction below for the sake of completeness.

\begin{figure}[ht!]
	\begin{mdframed}
	Let $\cG = (U,V,E,[R],\{\pi_e\}_{e \in E})$ be a $(1 - \epsilon,\delta)$-\ug~instance, where $\epsilon,\delta$ can be chosen as in the proof of Theorem \ref{thm:3-col}. Let $\cH = (V_{\cH},E_{\cH},w)$ be the weighted graph described in Figure \ref{fig:gadget-2}. Then we construct an instance $\cH' = (V',E',w')$ of Max-Cut as follows:
	
	{\bf Vertex Set}. Let $V' = V \times V_{\cH}$ i.e., for every Unique Game vertex $v \in V$, we introduce a copy of the vertex set of $\cH$.
	
	{\bf Edges}. The edges in $\cH$ are given using the following distribution:
	\begin{enumerate}
		\item Sample $u \sim U$ and two neighbors $v_1,v_2 \sim N_{\cG}(u)$ uniformly at random.
		\item Sample an edge $(x,y) \sim \cH$ using the distribution from Figure \ref{fig:gadget-2}.
		\item Output the edge $\{(v_1,x),(v_2,y)\}$.
	\end{enumerate} 
	\end{mdframed}
\end{figure}

Given Lemmas \ref{lem:comp-2} and \ref{lem:sound-2}, the completeness and soundness guarantees of the above reduction follow exactly as in the proof of Theorem \ref{thm:i-set}. We omit the details here. 

\subsection{Proof of Lemma \ref{lem:iset-dv1}}			\label{sec:iset-dv1}

We first observe that since the map $\mu_2 \to t(\mu_2)$ is invertible, we instead work with the function
\[
L(t) := \alpha G(A t) + 2(1 - \alpha)G(Bt).
\]
Under this change of variable, our task is to characterize the solutions to the constraint $L(t) = 0$. To that end, we observe the following:
\begin{itemize}
	\item Since $G(0) = 0$, it follows that $t = 0$ is a solution to $L(t) = 0$.
	\item Furthermore, note that since $G$ is an odd function, then if $t$ is a solution to $L(t) = 0$, then so is $-t$. Hence, we focus on analyzing $L(t)$ on the positive half $t > 0$.
\end{itemize}

Now, note that the constraint $L(t) = 0$ is equivalent to the constraint:
\[
\frac{G(|A|t)}{G(Bt)} = \frac{2(1 - \alpha)}{\alpha},
\] 
and hence we study the behavior of the ratio $Z(t) = \frac{G(|A|t)}{G(Bt)}$ instead.  To that end, we have the following claim:

\begin{claim}			\label{cl:Z-incr}
	The function $Z(t)$ is strictly increasing in $t$ for $t > 0$.
\end{claim}

\begin{proof}
We first note that
\[
Z'(t) = \frac{1}{G(Bt)^2}\left(|A|G'(|A|t)G(Bt) - |B|G'(Bt)G(|A|t)\right).
\]
and hence, showing $Z(t)$ is strictly increasing is equivalent to showing that
\[
|A|G'(|A|t)G(Bt) - |B|G'(Bt)G(|A|t) > 0
\]
or equivalently,
\begin{equation}		\label{eqn:G-eq}
	\frac{G(|A|t)}{G'(|A|t)} < \frac{|A|}{B}\cdot \frac{G(Bt)}{G'(Bt)}. 
\end{equation}
Hence, we will just focus on establishing \eqref{eqn:G-eq}. To that end, let us define the function
\[
h(x):= \frac{G(x)}{G'(x)} = \frac{\Phi(x) - \Phi(-x)}{2\phi(x)},
\]
and we claim that $h(x)$ is convex in $x$ for $x > 0$. i.e., we note:
\[
h'(x) = \frac{2 \phi(x)}{2 \phi(x)} + \frac{\Phi(x) - \Phi(-x)}{2 \phi(x)^2}\cdot \phi(x) \cdot x = 2 + x\cdot\frac{\Phi(x) - \Phi(-x)}{2\phi(x)},
\]
and 
\[
h''(x) = \frac{\Phi(x) - \Phi(-x)}{2\phi(x)} + \frac{x}{2\phi(x)} + x^2\frac{\Phi(x) - \Phi(-x)}{2\phi(x)},
\]
which is strictly positive for $x > 0$, thus implying that $h(x)$ is strictly convex. Then, observing that $h(0) = 0$ and using the convexity of $h$ we can now get that:
\[
h(|A|t) = h\left(\frac{|A|}{B}\cdot Bt\right) < \frac{|A|}{B}h(Bt),
\]
which establishes \eqref{eqn:G-eq}, which in turn implies the claim.
\end{proof}

Now, note that since $Z(t)$ is strictly increasing, it follows that there is at most one choice of $t = t_1 > 0$ for which 
\[
Z(t_1) = \frac{2(1 - \alpha)}{\alpha}.
\]
Moreover, for any small $\epsilon > 0$, we will also have that $Z(t_1 - \epsilon) < \frac{2 - 2\alpha}{\alpha} < Z(t_1 + \epsilon)$, i.e, or 
\[
\alpha G(A(t_1 - \epsilon)) + 2(1 - \alpha) G(B(t_1 - \epsilon)) > 0 > \alpha G(A(t_1 + \epsilon)) + 2(1 - \alpha)G(B(t_1 + \epsilon)),
\]
i.e., $t_1$ corresponds to a local maxima of $F$.

{\bf Finishing Up}. All that remains is to identify the cases when $t_1 > 0$ exists such that $Z(t_1) = \frac{2(1 - \alpha)}{\alpha}$. To that end, note that
\[
\lim_{t \to 0+} Z(t) = \frac{|A|}{B} = \frac{1 + \rho}{|\rho|}.
\]
Hence, we may consider two cases: 

{\bf Case (i)}. Suppose $|A|/|B| > 2(1 - \alpha)/\alpha$. Then $t = 0$ is the only solution to $L(t) = 0$. 

{\bf Case (ii)}. Suppose $|A|/|B| \leq 2(1 - \alpha)/\alpha$. Then there are exactly two roots for $L(t)$ other than $t = 0$, both of which are local maxima. 

Finally, as argued before, in both cases $t = 0$ corresponds to a local minima of $F_{\alpha}$.

\bibliographystyle{alpha}
\bibliography{main}

\appendix
\section{Auxiliary Proofs}

\subsection{Proof of Theorem \ref{thm:stab}}			\label{sec:thm-stab}

To derive the bound, we need to recall an additional notation from \cite{Mossel10}. For any choice of $\delta_1,\delta_2,\rho \in [0,1]$, let us define
\[
\lGamma_{\rho}(\mu_1,\mu_2) := \Pr_{g_1 \underset{\rho}{\sim} g_2} \Big[g_1 \leq \Phi^{-1}(\delta_1), g_2 \geq \Phi^{-1}(1 - \delta_2)\Big].
\]
Now let $(\Omega^2,\mu)$ and $f_1,f_2 \in L_2(\Omega)$ be as in the setting of Theorem \ref{thm:stab}. Then for any functions $f_1,f_2$ satisfying the premise of Theorem \ref{thm:stab}, Theorem 1.14 from \cite{Mossel10} states that
\[
\Ex_{(\omega_1,\omega_2) \sim \mu}\left[f_1(\omega_1)f_2(\omega_2)\right]
\geq \lGamma_\rho\big(\Ex [f_1], \Ex[f_2]\big) - c(\tau),
\]
where $\rho = \rho(\Omega^2,\mu)$. Now, denoting $\delta_i := \Ex[f_i]$ for $i = 1,2$,  we further manipulate the RHS as:
\begin{align*}
	\lGamma_\rho\big(\Ex[f_1],\Ex[f_2]\big)
	&= \Pr_{g_1 \underset{\rho}{\sim} g_2} \Big[g_1 \leq \Phi^{-1}(\delta_1), g_2 \geq \Phi^{-1}(1 - \delta_2)\Big] \\
	&= \Pr_{g_1 \underset{\rho}{\sim} g_2} \Big[g_1 \leq \Phi^{-1}(\delta_1), -g_2 \leq -\Phi^{-1}(1 - \delta_2)\Big] \\
	&\overset{1}{=} \Pr_{g_1 \underset{-\rho}{\sim} g'_2} \Big[g_1 \leq \Phi^{-1}(\delta_1), g'_2 \leq \Phi^{-1}(\delta_2)\Big] \\
	&= \Lambda_{-\rho}\big(\delta_1,\delta_2\big),
\end{align*}
where in step $1$ we use the following observations:
\begin{itemize}
	\item If a pair of standard normal variables $g,h$ are $\rho$-correlated, then the pair $g,-h$ are $-\rho$ correlated.
	\item Since $\phi(\cdot)$ is symmetric around the origin, it follows that $\Phi^{-1}(1 - x) = -\Phi^{-1}(x)$ for all $x \in [0,1]$.
\end{itemize}

\subsection{Proof of Lemma \ref{lem:dv}}			\label{sec:proof-dv}

We begin by deriving a more convenient-to-use expression for $\Lambda_\rho(x,y)$. In particular, note that since $g_2$ is a $\rho$-correlated copy of $g_1$, we can write $g_2 = \rho\cdot g_1 + \sqrt{1 - \rho^2}\cdot a$, where $a \sim N(0,1)$ is Gaussian independent of $g_1$. Hence, denoting $t = \Phi^{-1}$, we get:
\begin{align*}
	\Lambda_\rho(x,y)
	&= \Pr_{g_1 \underset{\rho}{\sim} g_2} \Big[g_1 \leq t(x), g_2 \leq t(y)\Big] \\
	&= \Pr_{g_1,a \sim N(0,1)}\Big[g_1 \leq t(x), \rho g_1 + \sqrt{1 - \rho^2} a \leq t(y)\Big] \\
	&= \Pr_{g_1,a \sim N(0,1)}\left[g_1 \leq t(x), a \leq \frac{t(y) - \rho g_1}{\sqrt{1 - \rho^2}}\right] \\
	&= \int^{t(x)}_{-\infty} \Phi\left(\frac{t(y) - \rho a}{\sqrt{1 - \rho^2}}\right) \phi(a) da.
\end{align*}
Using the above and applying the Leibniz integral rule we get that:
\begin{align*}
	\frac{\partial  \Lambda_\rho(x,y)}{\partial x}
	&= \frac{\partial}{\partial x} \int^{t(x)}_{-\infty} \Phi\left(\frac{t(y) - \rho a}{\sqrt{1 - \rho^2}}\right) \phi(a)da \\
	&= \frac{d t(x)}{d x} \cdot \Phi\left(\frac{t(y) - \rho t(x)}{\sqrt{1 - \rho^2}}\right) \phi(t(x)) \\
	&= \frac{1}{\phi(t(x))}\cdot \Phi\left(\frac{t(y) - \rho t(x)}{\sqrt{1 - \rho^2}}\right) \phi(t(x))	 \\
	&= \Phi\left(\frac{t(y) - \rho t(x)}{\sqrt{1 - \rho^2}}\right).
\end{align*}

\subsection{Influence Decoding Lemma}       \label{sec:dec}

Lemmas \ref{lem:inf-dec} and \ref{lem:inf-dec1} are a consequence of the following more general well-known lemma.

\begin{lemma}[Folklore, e.g,. Lemma 2.11 of~\cite{steurer-thesis}]			\label{lem:dec}
	Suppose $\cG = (U,V,E,[R],\{\pi_e\}_{e \in E})$ is a Unique Game instance with ${\sf Opt}(\cG) \leq \delta$. Let $\{f_v\}_{v \in V}$ be a collection of functions $f_v:\{0,1\}^R \to [0,1]$. Furthermore, for every $u \in U$, let $g_u$ be defined as
	\[
	g_u(x) = \Ex_{v \sim N_{\cG}(u)}\left[T_{1 - \eta}f_{v}\left(\pi_{v,u}(x)\right)\right].
	\]
	Then,
	\[
	\nu:= \Pr_{u \sim U}\left[\max_{j \in [R]} \Inf{j}{g_u} > \tau\right] \leq 8\frac{\sqrt{\delta}}{\tau^2\eta^2}.
	\]
\end{lemma}

Using the above lemma, we prove Lemma \ref{lem:inf-dec}.

\begin{proof}[Proof of Lemma \ref{lem:inf-dec}]
Let us denote
\[
\nu' = \Pr_{u \sim U}\Big[u \notin U_{\rm good}\Big].
\]
Then, by averaging, there exists $i \in \{1,2,3\}$ such that
\[
\Pr_{u \sim U}\Big[\max_{\ell \in [R]}\Inf{\ell}{g_{u,i}} > \tau\Big] \geq \frac{\nu'}{3}.
\]
Now instantiating $f_u := f_{u,i}$ and $g_{u} := g_{u,i}$ in the setting of Lemma \ref{lem:dec}, we get that 
\[
\frac{\nu'}{3} \leq \frac{8\sqrt{\delta}}{\tau^2\eta^2},
\]
from which the claim follows.
\end{proof}
Lemma \ref{lem:inf-dec1} can be established exactly as Lemma \ref{lem:inf-dec} and hence we omit the details here.

\section{Some Derivatives for Interval Arithmetic Verification}

In this appendix, we collect some derivatives that are used in Lemma~\ref{lem:ia}.

\begin{lemma}[Partial derivatives of $\rho(b_i, b_j)$]
    Let $b_i, b_j \in (-1, 1)$. We have the following partial derivatives of $\rho(b_i, b_j)$:
    \begin{enumerate}
        \item 
        \begin{equation}
            \frac{\partial}{\partial b_i} \rho(b_i, b_j) = \sqrt{\frac{1 - b_j}{1 + b_j}} \cdot \frac{1}{\sqrt{-(1 + b_i)^4 + 2(1 + b_i)^3}}
        \end{equation}
        \item 
        \begin{equation}
            \frac{\partial^2}{\partial b_i\partial b_j} \rho(b_i, b_j) = - \frac{1}{\sqrt{-(1 + b_i)^4 + 2(1 + b_i)^3}} \cdot \frac{1}{\sqrt{-(1 + b_j)^4 + 2(1 + b_j)^3}}
        \end{equation}
        \item 
        \begin{equation}
            \frac{\partial^2}{\partial b_i^2} \rho(b_i, b_j) = \sqrt{\frac{1 - b_j}{1 + b_j}} \cdot \frac{2(1 + b_i)^3 - 3(1 + b_i)^2}{\left(-(1 + b_i)^4 + 2(1 + b_i)^3\right)^{3/2}}
        \end{equation}
    \end{enumerate}
\end{lemma}
\begin{proof}
    Items (1) and (2) follow from the fact that
    \begin{align}
        \frac{\partial}{\partial b} \sqrt{\frac{1 - b}{1 + b}} & = \frac{1}{2}\cdot \sqrt{\frac{1+b}{1-b}} \cdot \frac{-2}{(1+b)^2} \\
        & = - \sqrt{\frac{1}{(1-b)(1+b)^3}}\\
        & = - \sqrt{\frac{1}{-(1+b)^4 + 2(1+b)^3}}.\\
    \end{align}
    Item 3 follows by taking derivative of the above in $b$ one more time.
\end{proof}

In the following we assume that $t$ is a piecewise-linear function and neither $b_i$ nor $b_j$ lie at any break point of $t$. We denote the slope of the piece containing $b_i$ by $c_i$. We also use $r_i$ as a shorthand for $\Phi^{-1}(t(b_i))$.

\begin{lemma}[Partial derivatives of $\varphi_\rho(r_i, r_j)$]
    We have
    \begin{equation}
        \frac{\partial}{\partial b_i}\varphi_\rho(r_i, r_j) = \varphi_\rho(r_i, r_j)\cdot \left(\frac{\partial \rho}{\partial b_i}\cdot\left(\frac{\rho}{1 - \rho^2} + \frac{r_ir_j(1-\rho^2) - \rho(r_i^2 - 2\rho r_ir_j + r_j^2)}{(1 - \rho^2)^2}\right) + \frac{c_i}{\varphi(r_i)} \cdot \frac{-(r_i - \rho r_j)}{1 - \rho^2}\right)
    \end{equation}
\end{lemma}

\begin{lemma}[Partial derivatives of $s(b_i, b_j)$]\label{lem:partial_s}
We have
\begin{enumerate}
\item
\begin{equation}
    \frac{\partial}{\partial b_i}s(b_i, b_j) = -2 \left(c_i \cdot \Phi\left(\frac{r_j - \rho r_i}{\sqrt{1 - \rho^2}}\right)+ \varphi_\rho(r_i, r_j) \cdot \frac{\partial\rho}{\partial b_i}\right)+ c_i
\end{equation}
\item 
\begin{align}
    \frac{\partial^2}{\partial b_i^2}s(b_i, b_j) & = -2\Bigg(c_i \cdot \varphi\left(\frac{r_j - \rho r_i}{\sqrt{1 - \rho^2}}\right)\cdot\left(\frac{-\rho \cdot \frac{c_i}{\varphi(r_i)} - \frac{\partial \rho}{\partial b_i} \cdot r_i}{\sqrt{1 - \rho^2}} + \frac{r_j - \rho r_i}{(1 - \rho^2)^{3/2}} \cdot \rho \cdot \frac{\partial \rho}{\partial b_i}\right) \\
    & \qquad + \frac{\partial^2\rho}{\partial b_i^2}\cdot \varphi_\rho(r_i, r_j) + \frac{\partial \rho}{\partial b_i}\cdot\frac{\partial}{\partial b_i}\varphi_\rho(r_i, r_j)\Bigg)
\end{align}
\item 
\begin{align}
    \frac{\partial^2}{\partial b_i\partial b_j}s(b_i, b_j) & = -2\Bigg(c_i \cdot \varphi\left(\frac{r_j - \rho r_i}{\sqrt{1 - \rho^2}}\right)\cdot\left(\frac{\frac{c_j}{\varphi(r_j)} - \frac{\partial \rho}{\partial b_j} \cdot r_i}{\sqrt{1 - \rho^2}} + \frac{r_j - \rho r_i}{(1 - \rho^2)^{3/2}} \cdot \rho \cdot \frac{\partial \rho}{\partial b_j}\right) \\
    & \qquad + \frac{\partial^2\rho}{\partial b_i\partial b_j}\cdot \varphi_\rho(r_i, r_j) + \frac{\partial \rho}{\partial b_i}\cdot\frac{\partial}{\partial b_j}\varphi_\rho(r_i, r_j)\Bigg)
\end{align}
\end{enumerate}
\end{lemma}

\end{document}